\documentclass[10pt,twocolumn,twoside]{IEEEtran}

\IEEEoverridecommandlockouts                              	

\usepackage{amssymb,amsmath,amsfonts,amsthm}
\usepackage{algorithm,algorithmic}
\usepackage{cite}
\usepackage{float}
\usepackage{epsfig}
\usepackage{graphicx}
\usepackage{graphics}
\usepackage{multicol}
\usepackage{lipsum}
\usepackage{subfigure}
\usepackage{url}
\usepackage{verbatim}
\usepackage{xspace}
\usepackage[usenames,dvipsnames]{color}
\usepackage{setspace}
\usepackage{dblfloatfix}

\floatstyle{ruled}
\newfloat{model}{H}{mod}
\floatname{model}{\footnotesize Model}
\newfloat{notatio}{H}{not}
\floatname{notatio}{\footnotesize Notation}

\newenvironment{varalgorithm}[1]
  {\algorithm\renewcommand{\thealgorithm}{#1}}
  {\endalgorithm}

\usepackage{url,changebar,bm,xspace,dsfont}
\let\mathbb=\mathds % I much prefer the dsfont over the bbfont
\newtheorem{prop}{Proposition}

\newtheorem{example}{\bfseries Example}

\newtheorem{remark}{Remark}

\ifCLASSINFOpdf
 \else
\fi

\usepackage{graphics} % for pdf, bitmapped graphics files
\usepackage{epsfig} % for postscript graphics files
\usepackage{algorithm,algorithmic}

\usepackage{amsmath,amsfonts,amssymb,amscd}
\usepackage{capt-of}
\usepackage{color}
\usepackage{cite}
\usepackage{verbatim}

\newtheorem{lemma}{\bfseries Lemma}

\begin{document}

\title{\LARGE \bf
Event-Triggered Quantized Average Consensus via Mass Summation}
%\author{ \parbox{3 in}{\centering Huibert Kwakernaak*
%         \thanks{*Use the $\backslash$thanks command to put 
%information here}\\
%         Faculty of Electrical Engineering, Mathematics and 
%Computer Science\\
%         University of Twente\\
%         7500 AE Enschede, The Netherlands\\
%         {\tt\small h.kwakernaak@autsubmit.com}}
%         \hspace*{ 0.5 in}
%         \parbox{3 in}{ \centering Pradeep Misra**
%         \thanks{**The footnote marks may be inserted manually}\\
%        Department of Electrical Engineering \\
%         Wright State University\\
%         Dayton, OH 45435, USA\\
%         {\tt\small pmisra@cs.wright.edu}}
%}

\author{Apostolos~I.~Rikos,~\IEEEmembership{Member,~IEEE}% <-this % stops a space
	\thanks{Apostolos~I.~Rikos is with Division of Decision and Control Systems, KTH Royal Institute of Technology, SE-100 44 Stockholm, Sweden. E-mail:{\tt~rikos@kth.se}.}, 
Christoforos~N.~Hadjicostis,~\IEEEmembership{Fellow,~IEEE} % <-this % stops a space 
	\thanks{Christoforos~N.~Hadjicostis is with the Department of Electrical and Computer Engineering at the University of Cyprus, Nicosia, Cyprus, and also with the Department of Electrical and Computer Engineering at the University of Illinois, Urbana-Champaign, IL, USA. E-mail:{\tt~chadjic@ucy.ac.cy}.}
	\thanks{Parts of the results for quantized average consensus via event-triggered mass summation appear in \cite{2018:RikosHadj}. The present version of the paper includes complete proofs for convergence as well as an enhanced version of the algorithm that allows nodes to determine when to seize transmissions, once quantized average consensus is reached (not addressed in \cite{2018:RikosHadj}).}
}

\maketitle
\thispagestyle{empty}
\pagestyle{empty}

\providecommand{\keywords}[1]{\textbf{\textit{Index terms---}} #1}

%%%%%%%%%%%%%%%%%%%%%%%%%%%%%%%%%%%%%%%%%%%%%%%%%%%%%%%%%%%%%%%%%%%%%%%%%%%%%%%%
\begin{abstract}

We study the distributed average consensus problem in multi-agent systems with directed communication links that are subject to quantized information flow. 
The goal of distributed average consensus is for the nodes, each associated with some initial value, to obtain the average (or some value close to the average) of these initial values. 
In this paper, we present and analyze novel distributed averaging algorithms which operate exclusively on quantized values (specifically, the information stored, processed and exchanged between neighboring agents is subject to deterministic uniform quantization) and rely on event-driven updates (e.g., to reduce energy consumption, communication bandwidth, network congestion, and/or processor usage). 
We characterize the properties of the proposed distributed averaging protocols on quantized values and show that their execution, on any time-invariant and strongly connected digraph, will allow all agents to reach, in finite time, a common consensus value represented as the ratio of two quantized values that is equal to the exact average. 
We conclude with examples that illustrate the operation, performance, and potential advantages of the proposed algorithms. 
\end{abstract}


\IEEEpeerreviewmaketitle

\vspace{-0.3cm}
% ===============================================
%
%
% INTRODUCTION
%
%
% ===============================================

\section{INTRODUCTION}\label{intro}

In recent years, there has been a growing interest for control and coordination of networks consisting of multiple agents, like groups of sensors \cite{2005:XiaoBoydLall} or mobile autonomous agents \cite{2004:Murray}. 
A problem of particular interest in distributed control is the \textit{consensus} problem where the objective is to develop distributed algorithms that can be used by a group of agents in order to reach agreement to a common decision. 
The agents start with different initial values/information and are allowed to communicate locally via inter-agent information exchange under some constraints on connectivity. 
Consensus processes play an important role in many problems, such as leader election \cite{1996:Lynch},  motion coordination of multi-vehicle systems \cite{2005:Olshevsky_Tsitsiklis, 2004:Murray}, and clock synchronization \cite{2007:Gamba}.

%The consensus problem also arises in a number of applications including coordination of UAVs (e.g., aligning the agents’ directions of motion), information processing in sensor networks, and distributed optimization (e.g., agreeing on the estimates of some unknown parameters). 
%The averaging problem is a special case in which the goal is to compute the exact average of the initial values of the agents.

One special case of the consensus problem is distributed averaging, where each agent (initially endowed with a numerical value) can send/receive information to/from other agents in its neighborhood and update its value iteratively, so that eventually, all agents compute the average of all initial values. 
Average consensus is an important problem and has been studied extensively in settings where each agent processes and transmits real-valued states with infinite precision \cite{2018:BOOK_Hadj, 2005:Olshevsky_Tsitsiklis, 2008:Sundaram_Hadjicostis, 2013:Themis_Hadj_Johansson, 2004:XiaoBoyd, 2010:Dimakis_Rabbat, 2011:Morse_Yu, 1984:Tsitsiklis}.

Most existing algorithms, for average consensus (and also consensus) provide asymptotic convergence to the consensus value and cannot be directly applied to real-world control and coordination applications.
For this reason there has been interest on finite time (average) consensus algorithms (e.g., \cite{2007:Shreyas_Hadjicostis, 2015:Charalambous_Hadjicostis, 2013:Yuana_Goncalves}) but the challenge of designing simple finite time algorithms for these tasks remains open. 
Furthermore, in practice, due to constraints on the bandwidth of communication links and the capacity of physical memories, both communication and computation need to be performed assuming finite precision.
For these reasons, researchers have also studied the case when network links can only allow messages of limited length to be transmitted between agents, effectively extending techniques for average consensus towards the direction of quantized consensus.
Various distributed strategies have been proposed, to allow the agents in a network to reach quantized consensus \cite{2007:Aysal_Rabbat, 2012:Lavaei_Murray, 2007:Basar, 2008:Carli_Zampieri, 2016:Chamie_Basar, 2011:Cai_Ishii}. 
Apart from \cite{2016:Chamie_Basar} (which converges in a deterministic fashion but requires a communication topology that forms a doubly stochastic matrix), these existing strategies use {\em randomized} approaches to address the quantized average consensus problem (implying that all agents reach quantized average consensus with probability one or in some other probabilistic sense); the design of {\em deterministic} distributed strategies that achieve quantized average consensus remains largely unexplored. 
An additional desirable feature in many types of communication networks is the infrequent update of values to avoid consuming valuable network resources. 
Thus, there has also been an increasing interest for novel event-triggered algorithms for distributed quantized average consensus (and, more generally, distributed control), in order to achieve more efficient usage of network resources \cite{2013:Dimarogonas_Johansson, 2016:nowzari_cortes, 2012:Liu_Chen}.

%The other direction adopts randomized
%time-varying networks with real-valued states, a model that
%potentially captures a variety of random behaviors exhibited
%in realistic networks [14], [18]–[23]; see also [24] for related
%problems in search engines. 

%In such case, one deals with an autonomous discrete-time linear system with a transition matrix $P$, also referred to as weight matrix, that is defined by the coefficients (weights) used in the linear updates. 
%It is well-known that if the weight matrix $P$ is primitive doubly stochastic, then the nodes will asymptotically converge to the average of their initial values (see, e.g., \cite{2004:XiaoBoyd} and references therein). 
%Choosing, in a distributed manner, weights that form a primitive doubly stochastic matrix (and
%can thus be used for asymptotic average consensus) is rather trivial in fixed interconnection topologies that are undirected; however, this is significantly more complex in the case of directed interconnection topologies (digraphs) \cite{2013:Christoforos, 2009:Cortes, 2010:Cortes, 2012:CortesJournal}. 
%An alternative approach, which avoids obtaining a doubly stochastic matrix altogether, is the so-called ratio-consensus \cite{2010:christoforos}, or push-sum algorithm \cite{2003:Kempe, 2012:TsianosLawlorRabbat}. 
%Ratio consensus relies on simultaneously running two iterations (each with different initial conditions) and has each node calculate the exact average by taking the ratio of its two iteration values.

In this paper, we present three novel distributed average consensus algorithms that combine the desirable features mentioned above. 
More specifically, average consensus is reached in finite time, and the processing, storing, and exchange of information between neighboring agents is subject to uniform quantization and ``event-driven''. 
Following \cite{2007:Basar, 2011:Cai_Ishii} we assume that the states are integer-valued (which comprises  a uniform class of quantization effects) and the control actuation of each node is event-based. 
We note that most work dealing with quantization has concentrated on the scenario where the agents can store and process real-valued states but can transmit only quantized values through limited rate channels (see, e.g., \cite{2008:Carli_Zampieri, 2016:Chamie_Basar}). 
By contrast, our assumption is also suited to the case where the states are stored in digital memories of finite capacity (as in \cite{2009:Nedic, 2007:Basar, 2011:Cai_Ishii}) as long as the initial values are also quantized. 
The paper establishes that the proposed algorithms allow all agents to reach quantized average consensus in finite time by reaching a value represented as the ratio of two integer values that is equal to the average. 
In the case of the probabilistic algorithm we present, this ratio equals the average in finite time with probability one.

The remainder of this paper is organized as follows. 
In Section~\ref{review} we review the existing literature related to our work while in Section~\ref{preliminaries}, we introduce the notation used throughout the paper.
In Section~\ref{probForm} we formulate the quantized average consensus problem.
In Section~\ref{rand_algorithm}, we present a probabilistic distributed algorithm, which allows the agents to reach consensus to the \textit{exact} quantized average of the initial values with probability one. 
In Section~\ref{DetAlgorithm}, we present a deterministic event-triggered version of the algorithm in Section~\ref{rand_algorithm}, and show that it reaches consensus to the \textit{exact} quantized average of the initial values after a finite number of steps, for which we also provide a worst case upper bound.  
In Section~\ref{MaxAlgorithm}, we present a deterministic event-triggered distributed algorithm, which, not only allows the agents to reach consensus to the \textit{exact} quantized average of the initial values after a finite number of steps, but also allows them to cease transmissions once quantized average consensus is reached. 
For each proposed algorithm, we analyze the operation and establish convergence to the quantized average of the initial values. 
In Section~\ref{results}, we present simulation results and comparisons. 
We conclude in Section~\ref{future} with a brief summary and remarks about future work.

% ===============================================
%
%
% LITERATURE REVIEW
%
%
% ===============================================
\section{LITERATURE REVIEW}\label{review}

In this section, we review existing literature on algorithms for distributed averaging under quantized communication, depending on whether they converge in a probabilistic or a deterministic fashion. 

In recent years, quite a few {\em probabilistic} distributed algorithms for averaging under quantized communication, have been proposed. 
Specifically, the probabilistic quantizer in \cite{2007:Aysal_Rabbat} converges to a common value with a random quantization level for the case where the topology forms a directed graph. 
In \cite{2010:Kar_Moura} the authors present a distributed algorithm which adds a dither over the agents' measurements (before the quantization process) and they show that the mean square error can be made arbitrarily small. 
In \cite{2011:Benezit_Vetterli} the authors present a distributed algorithm which guarantees all agents to reach a consensus value on the interval in which the average lies after a finite number of time steps. 
In \cite{2012:Lavaei_Murray} the authors present a quantized gossip algorithm which deals with the distributed averaging problem over a connected weighted graph, and calculate lower and upper bounds on the expected value of the convergence time, which depend on the principal submatrices of the Laplacian matrix of the weighted graph.

The available literature concerning {\em deterministic} distributed algorithms for averaging under quantized communication, comprises less publications. 
In \cite{2011:Li_Zhang} the authors present a distributed averaging algorithm with dynamic encoding and
decoding schemes. 
They show that for a connected undirected dynamic graph, average consensus is achieved asymptotically with
as few as one bit of information exchange between each pair of adjacent agents at each time step, and the 
convergence rate is asymptotic and depends on the number of network nodes, the number of quantization levels and the synchronizability of the network.
In \cite{2013:Thanou_Frossard} the authors present a novel quantization scheme for solving the average consensus problem when sensors exchange quantized state information. 
The proposed scheme is based on progressive reduction of the range of a uniform quantizer and it leads to progressive refinement of the information exchanged by the sensors.
In \cite{2008:Carli_Zampieri} the authors derive bounds on the rate of convergence to average consensus for a team of mobile agents exchanging information over time-invariant and randomly time-varying communication networks with symmetries. 
Furthermore, they study the control performance when agents also exchange logarithmically quantized data in static communication topologies with symmetries.
In \cite{2009:Nedic} the authors study distributed algorithms for the averaging problem over networks with time-varying topology, with a focus on tight bounds on the convergence time of a general class of averaging algorithms.
They consider algorithms for the case where agents can exchange and store continuous or quantized values, establish a tight convergence rate, and show that these algorithms guarantee convergence within some error from the average of the initial values; this error depends on the number of quantization levels.

Finally, recent papers have studied the quantized average consensus problem with the additional constraint that the state of each node is an integer value. 
In \cite{2007:Basar} the authors present a probabilistic algorithm which allows every agent to reach quantized consensus almost surely for a static and undirected communication topology, while in \cite{2016:Etesami_Basar} and \cite{2014:Basar_Olshevsky} they analyze and further improve its convergence rate. 
In \cite{2011:Cai_Ishii} a probabilistic algorithm was proposed to solve the quantized consensus problem for static directed graphs for the case where the agents exchange quantized information and store the changes of their states in an additional (also quantized) variable called `surplus'. 
In \cite{2016:Chamie_Basar} the authors present a deterministic distributed averaging protocol subject to quantization on the links and show that, depending on initial conditions, the system either converges in finite time to a quantized consensus, or the nodes' enter into a cyclic behaviour with their values oscillating around the average.

% ===============================================
%
%
% NOTATION
%
%
% ===============================================
\section{NOTATION AND MATHEMATICAL BACKGROUND}\label{preliminaries}

The sets of real, rational, integer and natural numbers are denoted by $ \mathbb{R}, \mathbb{Q}, \mathbb{Z}$ and $\mathbb{N}$, respectively. 
The symbol $\mathbb{Z}_+$ denotes the set of nonnegative integers.

%\subsection{Graph-Theoretic Notions}

Consider a network of $n$ ($n \geq 2$) agents communicating only with their immediate neighbors. 
The communication topology can be captured by a directed graph (digraph), called \textit{communication digraph}. 
A digraph is defined as $\mathcal{G}_d = (\mathcal{V}, \mathcal{E})$, where $\mathcal{V} =  \{v_1, v_2, \dots, v_n\}$ is the set of nodes (representing the agents) and $\mathcal{E} \subseteq \mathcal{V} \times \mathcal{V} - \{ (v_j, v_j) \ | \ v_j \in \mathcal{V} \}$ is the set of edges (self-edges excluded). 
A directed edge from node $v_i$ to node $v_j$ is denoted by $m_{ji} \triangleq (v_j, v_i) \in \mathcal{E}$, and captures the fact that node $v_j$ can receive information from node $v_i$ (but not the other way around). 
We assume that the given digraph $\mathcal{G}_d = (\mathcal{V}, \mathcal{E})$ is \textit{strongly connected} (i.e., for each pair of nodes $v_j, v_i \in \mathcal{V}$, $v_j \neq v_i$, there exists a directed \textit{path}\footnote{A directed \textit{path} from $v_i$ to $v_j$ exists if we can find a sequence of vertices $v_i \equiv v_{l_0},v_{l_1}, \dots, v_{l_t} \equiv v_j$ such that $(v_{l_{\tau+1}},v_{l_{\tau}}) \in \mathcal{E}$ for $ \tau = 0, 1, \dots , t-1$.} from $v_i$ to $v_j$). 
The subset of nodes that can directly transmit information to node $v_j$ is called the set of in-neighbors of $v_j$ and is represented by $\mathcal{N}_j^- = \{ v_i \in \mathcal{V} \; | \; (v_j,v_i)\in \mathcal{E}\}$, while the subset of nodes that can directly receive information from node $v_j$ is called the set of out-neighbors of $v_j$ and is represented by $\mathcal{N}_j^+ = \{ v_l \in \mathcal{V} \; | \; (v_l,v_j)\in \mathcal{E}\}$. 
The cardinality of $\mathcal{N}_j^-$ is called the \textit{in-degree} of $v_j$ and is denoted by $\mathcal{D}_j^-$ (i.e., $\mathcal{D}_j^- = | \mathcal{N}_j^- |$), while the cardinality of $\mathcal{N}_j^+$ is called the \textit{out-degree} of $v_j$ and is denoted by $\mathcal{D}_j^+$ (i.e., $\mathcal{D}_j^+ = | \mathcal{N}_j^+ |$).

We assume that each node is aware of its out-neighbors and can directly (or indirectly\footnote{Indirect transmission could involve broadcasting a message to all out-neighbors while including in the message header the ID of the out-neighbor it is intended for.}) transmit messages to each out-neighbor; however, it cannot necessarily receive messages (at least not directly) from them. 
In the randomized version of the protocol, each node $v_j$ assigns a nonzero \textit{probability} $b_{lj}$ to each of its outgoing edges $m_{lj}$ (including a virtual self-edge), where $v_l \in \mathcal{N}^+_j \cup \{ v_j \}$. 
This probability assignment can be captured by a column stochastic matrix $\mathcal{B} = [b_{lj}]$. 
A very simple choice\footnote{Note that this choice of nonzero probabilities is not unique. In fact, any positive values for the probabilities $b_{lj}$, for $v_l \in \mathcal{N}_j^+ \cup \{ v_j \}$, subject to the constraint that they sum to one, is also possible for the type of algorithms we discuss.} would be to set  
\begin{align*}
b_{lj} = \left\{ \begin{array}{ll}
         \frac{1}{1 + \mathcal{D}_j^+}, & \mbox{if $v_{l} \in \mathcal{N}_j^+ \cup \{v_j\}$,}\\
         0, & \mbox{otherwise.}\end{array} \right. 
\end{align*}
Each nonzero entry $b_{lj}$ of matrix $\mathcal{B}$ represents the probability of node $v_j$ transmitting towards the out-neighbor $v_l \in \mathcal{N}^+_j$ through the edge $m_{lj}$, or performing no transmission\footnote{From the definition of $\mathcal{B} = [b_{lj}]$ we have that $b_{jj} = \frac{1}{1 + \mathcal{D}_j^+}$, $\forall v_j \in \mathcal{V}$. This represents the probability that node $v_j$ will not perform a transmission to any of its out-neighbors $v_l \in \mathcal{N}^+_j$ (i.e., it will transmit to itself).}.

In the deterministic version of the protocol, each node $v_j$ assigns a \textit{unique order} in the set $\{0,1,..., \mathcal{D}_j^+ -1\}$ to each of its outgoing edges $m_{lj}$, where $v_l \in \mathcal{N}^+_j$. 
More specifically, the order of link $(v_l,v_j)$ for node $v_j$ is denoted by $P_{lj}$ (such that $\{P_{lj} \; | \; v_l \in \mathcal{N}^+_j\} = \{0,1,..., \mathcal{D}_j^+ -1\}$). 
This unique predetermined order is used during the execution of the proposed distributed algorithm as a way of allowing node $v_j$ to transmit messages to its out-neighbors in a \textit{round-robin}\footnote{When executing the deterministic protocol, each node $v_j$ transmits to its out-neighbors, one at a time, by following a predetermined order. The next time it transmits to an out-neighbor, it continues from the outgoing edge it stopped the previous time and cycles through the edges in a round-robin fashion according to the predetermined ordering.} fashion.

\section{PROBLEM FORMULATION}\label{probForm}

Consider a strongly connected digraph $\mathcal{G}_d = (\mathcal{V}, \mathcal{E})$, where each node $v_j \in \mathcal{V}$ has an initial (i.e., for $k=0$) quantized value $y_j[0]$ (for simplicity, we take $y_j[0] \in \mathbb{Z}$). 
In this paper, we develop a distributed algorithm that allows nodes (while processing and transmitting \textit{quantized} information via available communication links between nodes) to eventually obtain, after a finite number of steps, a fraction $q^s$ which is equal to the average $q$ of the initial values of the nodes, where
\begin{equation}
q = \frac{\sum_{l=1}^{n}{y_l[0]}}{n} .
\end{equation}

\begin{remark}
Following \cite{2007:Basar, 2011:Cai_Ishii} we assume that the state of each node is integer valued. 
This abstraction subsumes a class of quantization effects (e.g., uniform quantization).
\end{remark}

The algorithms we develop are iterative. 
With respect to quantization of information flow, we have that at time step $k \in \mathbb{Z}_+$ (where $\mathbb{Z}_+$ is the set of nonnegative integers), each node $v_j \in \mathcal{V}$ maintains the state variables $y^s_j, z^s_j, q_j^s$, where $y^s_j \in \mathbb{Z}$, $z^s_j \in \mathbb{Z}_+$ and  $q_j^s = \frac{y_j^s}{z_j^s}$, and the mass variables $y_j, z_j$, where $y_j \in \mathbb{Z}$ and $z_j \in \mathbb{Z}_+$.
The aggregate states are denoted by $y^s[k] = [y^s_1[k] \ ... \ y^s_n[k]]^{\rm T} \in \mathbb{Z}^n$, $z^s[k] = [z^s_1[k] \ ... \ z^s_n[k]]^{\rm T} \in \mathbb{Z}_+^n$, $q^s[k] = [q^s_1[k] \ ... \ q^s_n[k]]^{\rm T} \in \mathbb{Q}^n$ and $y[k] = [y_1[k] \ ... \ y_n[k]]^{\rm T} \in \mathbb{Z}^n$, $z[k] = [z_1[k] \ ... \ z_n[k]]^{\rm T} \in \mathbb{Z}_+^n$ respectively. 

Following the execution of the proposed distributed algorithms, we argue that there exists $\ k_0$ so that for every $k \geq k_0$ we have 
\begin{equation}\label{alpha_z_y}
y^s_j[k] = \frac{\sum_{l=1}^{n}{y_l[0]}}{\alpha}  \ \ \text{and} \ \ z^s_j[k] = \frac{n}{\alpha} ,
\end{equation}
where $\alpha \in \mathbb{N}$. This means that 
\begin{equation}\label{alpha_q}
q^s_j[k] = \frac{(\sum_{l=1}^{n}{y_l[0]}) / \alpha}{n / \alpha} = q ,
\end{equation}
for every $v_j \in \mathcal{V}$ (i.e., for $k \geq k_0$ every node $v_j$ has calculated $q$ as the ratio of two integer values).

%Consider a strongly connected digraph $\mathcal{G}_d = (\mathcal{V}, \mathcal{E})$, where each node $v_j \in \mathcal{V}$ has initial value (i.e., for $k=0$) mass variables $y_j[0], z_j[0]$ where $y_j[0] \in \mathbb{Z}$ and $z_j[0] = 1$, and initial state variables $y^s_j[0], z^s_j[0], q_j^s[0]$, where $y^s_j[0] = y_j[0]$, $z^s_j[0] = z_j[0]$, and $q_j^s[0] = \frac{y_j^s[0]}{z_j^s[0]}$. 
%In this paper we develop a distributed algorithm that allows nodes to process and transmit \textit{quantized} information via available (pairwise) communication links between nodes, so that the nodes eventually obtain, after a finite number of steps, a quantized fraction $q^s$ which is equal to the average $q$ of the initial values, where
%\begin{equation}
%q = \frac{\sum_{j=1}^{n}{y_j[0]}}{n} .
%\end{equation}

% ===============================================
%
%
% ALGORITHM
%
%
% ===============================================

\section{RANDOMIZED QUANTIZED AVERAGING WITH MASS SUMMATION}\label{rand_algorithm}

In this section we propose a distributed information exchange process in which the nodes, each having an integer initial value, transmit and receive quantized (integer) messages so that they reach average consensus on their initial values after a finite number of steps.

\subsection{Randomized Distributed Algorithm with Mass Summation}

\noindent
The operation of the proposed distributed algorithm is summarized below.

\noindent
\textbf{Initialization:}
Each node $v_j$ selects a set of probabilities $\{ b_{lj} \ | \ v_{l} \in \mathcal{N}_j^+ \cup \{v_j\} \}$ such that $0 < b_{lj} < 1$ and $\sum_{v_{l} \in \mathcal{N}_j^+ \cup \{v_j\}} b_{lj} = 1$ (see Section~\ref{preliminaries}). 
Each value $b_{lj}$, represents the probability for node $v_j$ to transmit towards out-neighbor $v_l \in \mathcal{N}^+_j$ (or perform a self transmission), at any given time step (independently between time steps and between nodes).  
Each node has some initial value $y_j[0]$, and also sets its state variables, for time step $k=0$, as $z_j[0] = 1$, $z^s_j[0] = 1$ and $y^s_j[0] = y_j[0]$, which means that $q^s_j[0] = y_j[0] / 1$. 
%Then, according to the nonzero probability $b_{lj}$, node $v_j$ either transmits $z_j[0]$ and $y_j[0]$ towards an out-neighbor $v_l \in \mathcal{N}_j^+$ or performs no transmission. 
%If it performed a transmission towards an out-neighbor, it sets $y_j[0] = 0$ and $z_j[0] = 0$. 

\noindent
The iteration involves the following steps:

\noindent
\textbf{Step 1. Transmitting:} According to the nonzero probabilities $b_{lj}$, assigned by node $v_j$ during the initialization step, it either transmits $z_j[k]$ and $y_j[k]$ towards a randomly selected out-neighbor $v_l \in \mathcal{N}_j^+$ or performs a self transmission. 
If it performs a transmission towards an out-neighbor $v_l \in \mathcal{N}_j^+$, it sets $y_j[k] = 0$ and $z_j[k] = 0$.

\noindent
\textbf{Step 2. Receiving:} Each node $v_j$ may receive messages $y_i[k]$ and $z_i[k]$ from its in-neighbor $v_i \in \mathcal{N}_j^-$ or itself; it sums all such messages it receives (if any) along with its stored mass variables $y_j[k]$ and $z_j[k]$ as
$$
y_j[k+1] = y_j[k] + \sum_{v_i \in \mathcal{N}_j^-} w_{ji}[k]y_i[k] ,
$$
and 
$$
z_j[k+1] = z_j[k] + \sum_{v_i \in \mathcal{N}_j^-} w_{ji}[k]z_i[k] ,
$$
where $w_{ji}[k] = 0$ (or $w_{jj}[k] = 0$) if no message is received from in-neighbor $v_i \in \mathcal{N}_j^-$; otherwise $w_{ji}[k] = 1$.

\noindent
\textbf{Step 3. Processing:} 
If $z_j[k+1] \geq z^s_j[k]$, node $v_j$ sets $z^s_j[k+1] = z_j[k+1]$, $y^s_j[k+1] = y_j[k+1]$ and 
$$
q^s_j[k+1] = \frac{y^s_j[k+1]}{z^s_j[k+1]} .
$$
Then, $k$ is set to $k+1$ and the iteration repeats (it goes back to Step~1).

% \subsection{Description of Randomized Distributed Mass Summation Algorithm}

The proposed algorithm is essentially a probabilistic quantized mass transfer process and is detailed as Algorithm~\ref{algorithm_prob} below (for the case when $b_{lj} = 1/(1+\mathcal{D}_j^+)$ for $v_l \in \mathcal{N}_j^+ \cup \{ v_j \}$ and $b_{lj}=0$ otherwise). 
Due to space limitations we do not illustrate the operation of the proposed algorithm, however, an analytical illustration can be found in \cite{2018:RikosHadj}.

\noindent
\vspace{-0.5cm}    
\begin{varalgorithm}{1}
\caption{Probabilistic Quantized Average Consensus}
\textbf{Input} 
\\ 1) A strongly connected digraph $\mathcal{G}_d = (\mathcal{V}, \mathcal{E})$ with $n=|\mathcal{V}|$ nodes and $m=|\mathcal{E}|$ edges. 
\\ 2) For every $v_j \in \mathcal{V}$ we have $y_j[0] \in \mathbb{Z}$. 
\\
\textbf{Initialization} 
\\ Every node $v_j \in \mathcal{V}$ does the following:
\\ 1) It assigns a nonzero probability $b_{lj}$ to each of its outgoing edges $m_{lj}$ and its self-edge, where $v_l \in \mathcal{N}^+_j \cup \{v_j\}$, as follows  
\begin{align*}
b_{lj} = \left\{ \begin{array}{ll}
         \frac{1}{1 + \mathcal{D}_j^+}, & \mbox{if $l = j$ or $v_{l} \in \mathcal{N}_j^+$,}\\
         0, & \mbox{if $l \neq j$ and $v_{l} \notin \mathcal{N}_j^+$.}\end{array} \right. 
\end{align*}
\\ 2) It sets $z_j[0] = 1$, $z^s_j[0] = 1$ and $y^s_j[0] = y_j[0]$ (which means that $q^s_j[0] = y_j[0] / 1$). 
\\
\textbf{Iteration}
\\ For $k=0,1,2,\dots$, each node $v_j \in \mathcal{V}$ does the following:
\\ 1) It either transmits $y_j[k]$ and $z_j[k]$ towards a randomly chosen out-neighbor $v_l \in \mathcal{N}_j^+$ (according to the nonzero probability $b_{lj}$) or performs a self transmission (according to the nonzero probability $b_{jj}$). If it transmitted towards an out-neighbor, it sets $y_j[k] = 0$ and $z_j[k] = 0$.
\\ 2) It receives $y_i[k]$ and $z_i[k]$ from its in-neighbors $v_i \in \mathcal{N}_j^-$ and sets 
$$
y_j[k+1] = y_j[k] + \sum_{v_i \in \mathcal{N}_j^-} w_{ji}[k]y_i[k] ,
$$
and 
$$
z_j[k+1] = z_j[k] + \sum_{v_i \in \mathcal{N}_j^-} w_{ji}[k]z_i[k] ,
$$
where $w_{ji}[k] = 1$ if node $v_j$ receives values from node $v_i$ (otherwise $w_{ji}[k] = 0$). 
\\ 3) If the following condition holds,
\begin{equation}
z_j[k+1] \geq z^s_j[k],
\end{equation}
it sets $z^s_j[k+1] = z_j[k+1]$, $y^s_j[k+1] = y_j[k+1]$, which means that $q^s_j[k+1] = \frac{y^s_j[k+1]}{z^s_j[k+1]}$. 
\\ 4) It repeats (increases $k$ to $k + 1$ and goes back to Step~1).
\label{algorithm_prob}
\end{varalgorithm}

\begin{remark}
From the operation of Algorithm~\ref{algorithm_prob}, it is important to notice that, once the initial mass variables ``merge'' (i.e., Step~$2$ of the Iteration of Algorithm~\ref{algorithm_prob}), they \textit{remain} ``merged'' during the operation of Algorithm~\ref{algorithm_prob}. 
\end{remark}

\subsection{Finite Time Convergence Analysis}

We are now ready to prove that during the operation of Algorithm~\ref{algorithm_prob} each agent obtains two integer values $y^s$ and $z^s$, the ratio of which is equal to the average $q$ of the initial values of the nodes. 

\begin{prop}
\label{PROP1_prob}
Consider a strongly connected digraph $\mathcal{G}_d = (\mathcal{V}, \mathcal{E})$ with $n=|\mathcal{V}|$ nodes and $m=|\mathcal{E}|$ edges, and $z_j[0] = 1$ and $y_j[0] \in \mathbb{Z}$ for every node $v_j \in \mathcal{V}$ at time step $k=0$. 
Suppose that each node $v_j \in \mathcal{V}$ follows the Initialization and Iteration steps as described in Algorithm~\ref{algorithm_prob}. 
Let $\mathcal{V}^+[k] \subseteq \mathcal{V}$ be the set of nodes $v_j$ with positive mass variable $z_j[k]$ at iteration $k$ (i.e., $\mathcal{V}^+[k] = \{ v_j \in \mathcal{V} \; | \; z_j[k] > 0 \}$). 
During the execution of Algorithm~\ref{algorithm_prob}, for every $k \geq 0$, we have that
$$
1 \leq | \mathcal{V}^+[k+1] | \leq | \mathcal{V}^+[k] | \leq n .
$$
\end{prop}

\begin{proof}
Steps~$1$ and $2$ at iteration $k$ of Algorithm~\ref{algorithm_prob} can be expressed according to the following equations 
\begin{eqnarray}\label{y_exch}
y[k+1] = W[k] \ y[k] ,
\end{eqnarray}
\begin{eqnarray}\label{z_exch}
z[k+1] = W[k] \ z[k] ,
\end{eqnarray}
where $y[k] = [y_1[k] \ ... \ y_n[k]]^{\rm T}$, $z[k] = [z_1[k] \ ... \ z_n[k]]^{\rm T}$ and $W[k] = [w_{lj}[k]]$ is an $n \times n$ binary column stochastic matrix. 
More specifically, for every $k$, the weights $w_{lj}[k]$, for $l=j$ or $l$ such that $(v_l,v_j) \in \mathcal{E}$, are either equal to $1$ or $0$, and furthermore each column sums to one.

Focusing on (\ref{z_exch}), at time step $k_0$, let us assume without loss of generality that
$z[k_0] = \left [ z_1[k_0] \dots z_{p_0}[k_0] \; 0 \dots \; 0 \right ]^{\rm T}$, where $z_i[k_0] > 0$, $\forall \ v_i \in \{ v_1, \cdots, v_{p_0} \}$ and $z_l[k_0] = 0$, $\forall \ v_l \in \mathcal{V} - \{ v_1, \cdots, v_{p_0} \}$. 
We can assume without loss of generality that the nodes with zero mass do not transmit (or transmit to themselves). 
Let us consider the scenario where $\sum_{v_i \in \mathcal{N}_j^- \cup \{v_j\}} w_{ji}[k_0] = 1$, $\forall \ v_j \in \mathcal{V}$ (i.e., for every row of $W[k_0]$ \textit{exactly} one element is equal to $1$ and all the other elements are equal to zero). 
This means that each node $v_j$ will receive at most one mass variable $z_i[k_0]$ and, since, at time step $k_0$, we have $p_0$ nodes with nonzero mass variables, we have that at time step $k_0 + 1$, exactly $p_0$ nodes have a nonzero mass variable. 
As a result, for this scenario, we have $| \mathcal{V}^+[k_0+1] | = | \mathcal{V}^+[k_0] |$.

Let us now consider the scenario where $w_{ji_1}[k_0] = 1$, $w_{ji_2}[k_0] = 1$ (where $v_{i_1}, v_{i_2} \ \in ( \mathcal{N}_j^- \cup \{ v_j \} ) \cap \mathcal{V}^+[k_0] $) and $w_{ji}[k_0] = 0, \forall \ v_i \in ( \mathcal{N}_j^- \cup \{ v_j \} ) - \{ v_{i_1}, v_{i_2} \}$ (i.e., the $j^{th}$ row of matrix $\mathcal{W}[k_0]$ has exactly $2$ elements equal to $1$ and all the other elements zero). 
Also, let us assume that $\sum_{v_i \in \mathcal{N}_l^- \cup \{v_l\}} w_{li}[k_0] \leq 1$, $\forall \ v_l \in \mathcal{V} - \{ v_j \}$ (i.e., for every row of $W[k_0]$ (except row $j$) \textit{at most} one element is equal to $1$ and all the other elements are equal to zero). 
The above assumptions, regarding matrix $W[k]$, mean that, during time step $k_0$, only node $v_j$ will receive two mass variables (from nodes $v_{i_1}$ and $v_{i_2}$) and all the other nodes will receive at most one mass variable. 
We have that $z_{j}[k_0 + 1] = z_{i_1}[k_0] + z_{i_2}[k_0]$ and $z_{l}[k_0 + 1] = z_{i(l)}[k_0]$, for $v_l \in \mathcal{V} - \{ v_j \}$ and some $v_{i(l)} \in \mathcal{V} - \{ v_{i_1}, v_{i_2} \}$ (i.e., node $v_j$ received two nonzero mass variables while all the other nodes received at most one nonzero mass variable, also counting its own mass variables). 
Since, at time step $k_0$, we had $p_0$ nodes with nonzero mass variables and at time step $k_0 + 1$ node $v_j$ received (and summed) two nonzero mass variables, while all the other nodes received at most one nonzero mass variable, this means that, at time step $k_0 + 1$, we have $p_0 - 1$ nodes with nonzero mass variables. 
This means that $| \mathcal{V}^+[k_0+1] | < | \mathcal{V}^+[k_0] |$.

%\noindent
%For case (b) we have that $z_{i_2}[k_0] = 0$ which means that node $v_j$ received only one nonzero mass variable. 
%As a result we have $z_{j}[k_0 + 1] = z_{i_1}[k_0]$ and $z_{l}[k_0 + 1] = z_{i}[k_0]$, for $v_l \in \mathcal{V} - \{ v_j \}$ and some $v_i \in \mathcal{V} - \{ v_{i_1}, v_{i_2} \}$ (i.e., every node (including $v_j$) received at most one nonzero mass variable). 
%Since every node transmitted to only one out-neighbor (Step~1 of Algorithm~\ref{algorithm_prob}), and every node received at most one nonzero mass variable, we have that $p_0$ nodes have a nonzero mass variable during time step $k_0 + 1$. 
%This means that $\mathcal{V}^+[k_0+1] = \mathcal{V}^+[k_0]$.
%
%\noindent
%For cases (c) and (d) we have $\mathcal{V}^+[k_0+1] = \mathcal{V}^+[k_0]$ (which can be shown in a similar way as (b)). 

By extending the above analysis for scenarios where each row of $W[k]$, at different time steps $k$, may have multiple elements equal to $1$ (but $W[k]$ remains column stochastic), we can see that the number of nodes $v_j$ with nonzero mass variable $z_{j}[k] > 0$ is non-increasing and thus we have $| \mathcal{V}^+[k+1] | \leq | \mathcal{V}^+[k] |$, $\forall \ k \in \mathbb{Z}_+$. 
\end{proof}

\begin{prop}
\label{PROP2_prob}
Consider a strongly connected digraph $\mathcal{G}_d = (\mathcal{V}, \mathcal{E})$ with $n=|\mathcal{V}|$ nodes and $m=|\mathcal{E}|$ edges and $z_j[0] = 1$ and $y_j[0] \in \mathbb{Z}$ for every node $v_j \in \mathcal{V}$ at time step $k=0$. 
Suppose that each node $v_j \in \mathcal{V}$ follows the Initialization and Iteration steps as described in Algorithm~\ref{algorithm_prob}. 
With probability one, we can find $k_0 \in \mathbb{Z}_+$, so that for every $k \geq k_0$ we have 
$$
y^s_j[k] = \sum_{l=1}^{n}{y_l[0]}  \ \ \text{and} \ \ z^s_j[k] = n , \ \ \forall v_j \in \mathcal{V}
$$
which means that 
$$
q^s_j[k] = \frac{\sum_{l=1}^{n}{y_l[0]}}{n} ,
$$
for every $v_j \in \mathcal{V}$ (i.e., for $k \geq k_0$ every node $v_j$ has calculated $q$ as the ratio of two integer values). 
\end{prop}

\begin{proof}
From Proposition~\ref{PROP1_prob} we have that $| \mathcal{V}^+[k+1] | \leq | \mathcal{V}^+[k] |$ (i.e., the number of nonzero mass variables is non-increasing). 
We will first show that the number of nonzero mass variables is decreasing after a finite number of steps, until, at some $k'_0 \in \mathbb{Z}_+$, we have $y_j[k'_0] =\sum_{l=1}^{n}{y_l[0]}  \ \ \text{and} \ \ z_j[k'_0] = n$, for some node $v_j \in \mathcal{V}$, and $y_i[k'_0] = 0  \ \ \text{and} \ \ z_i[k'_0] = 0$, for each $ v_i \in \mathcal{V} - \{ v_j \}$).

We have that Steps~$1$ and $2$ at iteration (time step) $k$ can be expressed according to (\ref{y_exch}) and (\ref{z_exch}).
Focusing on (\ref{z_exch}), consider for example, two nodes $v_i$ and $v_j$ that happen to share a common out-neighbor (say $v_l$): suppose that, during time step $k_0$, we have $z_i[k_0] > 0$, $z_j [k_0] > 0$  and $w_{li}[k_0] = 1$, $w_{lj} [k_0] = 1$. 
This scenario will occur with probability equal to $(1 + \mathcal{D}_i^+)^{-1}(1 + \mathcal{D}_j^+)^{-1}$ (i.e., as long as nodes $v_i$ and $v_j$ both transmit towards node $v_l$). 
Of course, for this to happen we need to have node $v_l$ be a common neighbor to nodes $v_i$ and $v_j$.
More generally, since the graph is strongly connected, for any pair of nodes $v_i$ and $v_j$, we can find a node (say $v_l$) and two paths (of length at most $n-1$) such that the first path $p_{li}$ connects $v_i$ to $v_l$ and the second path $p_{lj}$ connects $v_j$ to $v_l$. 
If the two paths are not of equal length, we can make them of equal length (at most $n-1$) by inserting one (or more) self loops in the shortest of the two paths ($p_{li}$ or $p_{lj}$). 
Then, it is easy to see that if, during time step $k_0$, we have $z_i[k_0] > 0$, $z_j[k_0] > 0$ (for any two nodes, $v_i$ and $v_j$, $i \neq j$), the two masses will merge at some node $v_l$ after at most $n-1$ steps, with probability 
\begin{eqnarray}\label{two_to_all}
\text{P}_{\text{two merge}} & = & \big( \prod_{v_{j'} \in p_{lj}}(1 + \mathcal{D}_{j'}^+)^{-1} \big) \big( \prod_{v_{i'} \in p_{li}}(1 + \mathcal{D}_{i'}^+)^{-1} \big) \nonumber \\
 & \geq & \big( \prod_{v_{j'} \in p_{lj}}(1 + \mathcal{D}_{max}^+)^{-1} \big) \big( \prod_{v_{i'} \in p_{li}}(1 + \mathcal{D}_{max}^+)^{-1} \big) \nonumber
 \\
 & \geq & \big((1 + \mathcal{D}_{max}^+)^{-1} \big)^{2(n-1)} \; ,
\end{eqnarray}
where $\mathcal{D}_{max}^+ = \max_{v_j \in \mathcal{V}}{\mathcal{D}_j^+}$ and $| p_{lj} | = | p_{li} |$ (since by inserting a sufficient number of self loops in the (shorter of the two) paths we can make the lengths of both paths $p_{lj}$ and $p_{li}$ equal (at most) to $n-1$). 
Note that the notation $v_{j'} \in p_{lj}$ means that there exists a directed path $p_{lj}$ consisting of a $v_j \equiv v_{l_0},v_{l_1}, \dots, v_{l_t} \equiv v_l$ (such that $(v_{l_{\tau+1}},v_{l_{\tau}}) \in \mathcal{E}$ for $ \tau = 0, 1, \dots , t-1$ i.e., a directed path from $v_j$ to $v_l$), and $v_{j'} \in \{ v_{l_0},v_{l_1}, \dots, v_{l_t} \}$. 
Note that (\ref{two_to_all}) provides a lower bound on the probability that, every $n$ time steps, two (or more) masses merge into one mass.

By extending the above discussion, we have that after $k = \tau (n-1)$ time steps (i.e., $\tau$ ``windows'', $\tau \geq (n-1)$, each consisting of $n-1$ time steps) the probability that all $n$ masses will ``merge'' into one mass is
\begin{eqnarray}\label{one_mass}
\text{P}_{\text{single mass}} & \geq & 1 - \sum_{j=0}^{n-2} {\tau \choose j} \big(1 - \text{P}_{\text{two merge}} \big)^{\tau - j} \ \text{P}_{\text{two merge}}^j \; , \nonumber
\end{eqnarray}
(where the summation on the right is an upper bound on the probability that $n-2$ or less mergings occur over the $\tau$ windows of length $n-1$.

Thus, by executing Algorithm~\ref{algorithm_prob} for $\tau$ ``windows'' (each consisting of $n-1$ time steps), we have that 
\begin{eqnarray}
\lim_{\tau \rightarrow \infty} \text{P}_{\text{single mass}} & = & 1 \; . \nonumber
\end{eqnarray}

\noindent
This means that, with probability one, $\exists k'_0 \in \mathbb{Z}_+$ for which $y_j[k'_0] =\sum_{l=1}^{n}{y_l[0]}  \ \ \text{and} \ \ z_j[k'_0] = n$, for some node $v_j \in \mathcal{V}$, and $y_i[k'_0] = 0,  \ \ \text{and} \ \ z_i[k'_0] = 0$, for each $ v_i \in \mathcal{V} - \{ v_j \}$. 
Once this ``merging'' of all nonzero mass variables occurs, we have that the nonzero mass variables of node $v_j$ will update the state variables of every node $v_i \in \mathcal{V}$ (because it will eventually be forwarded to all other nodes), which means that $\exists k_0 \in \mathbb{Z}_+$ (where $k_0 > k'_0$) for which $y_i^s[k_0] =\sum_{l=1}^{n}{y_l[0]}  \ \ \text{and} \ \ z_i^s[k_0] = n$, for every node $v_i \in \mathcal{V}$. 
Therefore, after a finite number of steps, (\ref{alpha_z_y}) and (\ref{alpha_q}) will hold for every node $v_j \in \mathcal{V}$ for the case where $\alpha = 1$. 
\end{proof}

% ===============================================
%
%
% ALGORITHM DETERMINISTIC EVENT TRIGGERED
%
%
% ===============================================
\section{EVENT-TRIGGERED QUANTIZED AVERAGING ALGORITHM WITH MASS SUMMATION}\label{DetAlgorithm}

In this section we propose a distributed algorithm in which the nodes receive quantized messages and perform transmissions according to a set of deterministic \textit{conditions}, so that they reach quantized average consensus on their initial values. 
This allows the calculation of an explicit worst-case upper bound regarding the number of steps required for quantized consensus. 
Unlike the operation of Algorithm~\ref{algorithm_prob} where, after a finite number of steps $k_0$, (\ref{alpha_z_y}) and (\ref{alpha_q}) will hold for each node $v_j$ with $\alpha = 1$ (at least with high probability), we will see that $\alpha$ can be (under some rare circumstances) an integer larger than $1$ in the deterministic algorithm of this section.

\subsection{Event-Triggered Deterministic Distributed Algorithm with Mass Summation}

The operation of the proposed distributed algorithm is summarized below.

\noindent
\textbf{Initialization:}
Each node $v_j$ assigns to each of its outgoing edges $v_l \in \mathcal{N}^+_j$ a \textit{unique order} $P_{lj}$ in the set $\{0,1,..., \mathcal{D}_j^+ -1\}$, which will be used to transmit messages to its out-neighbors in a round-robin fashion. 
Node $v_j$ has initial value $y_j[0]$ and sets its state variables, for time step $k=0$, as $z_j[0] = 1$, $z^s_j[0] = 1$ and $y^s_j[0] = y_j[0]$, which means that $q^s_j[0] = y_j[0] / 1$.
Then, it chooses an out-neighbor $v_l \in \mathcal{N}_j^+$ (according to the predetermined order $P_{lj}$) and transmits $z_j[0]$ and $y_j[0]$ to that particular neighbor. 
Then, it sets $y_j[0] = 0$ and $z_j[0] = 0$ (since it performed a transmission).

\noindent
The iteration involves the following steps:

\noindent
\textbf{Step 1. Receiving:} Each node $v_j$ receives messages $y_i[k]$ and $z_i[k]$ from its in-neighbors $v_i \in \mathcal{N}_j^-$ and sums them to obtain
$$
y_j[k+1] = y_j[k] + \sum_{v_i \in \mathcal{N}_j^-} w_{ji}[k]y_i[k] ,
$$
and 
$$
z_j[k+1] = z_j[k] + \sum_{v_i \in \mathcal{N}_j^-} w_{ji}[k]z_i[k] ,
$$
where $w_{ji}[k] = 0$ if no message is received from in-neighbor $v_i \in \mathcal{N}_j^-$; otherwise $w_{ji}[k] = 1$. 

\noindent
\textbf{Step 2. Event Trigger Conditions:} Node $v_j$ checks the following conditions:
\begin{enumerate}
\item It checks whether $z_j[k+1]$ is greater than $z^s_j[k]$.
\item If $z_j[k+1]$ is equal to $z^s_j[k]$, it checks whether $y_j[k+1]$ is greater than or equal to $y^s_j[k]$.
\end{enumerate}
If one of the above two conditions holds, it sets $y^s_j[k + 1] = y_j[k+1]$, $z^s_j[k + 1] = z_j[k+1]$ and $q_j^s[k+1] = \frac{y_j^s[k+1]}{z_j^s[k+1]}$. 

\noindent
\textbf{Step 3. Transmitting:} If the ``Event Trigger Conditions'' above do not hold, no transmission is performed. 
Otherwise, if the ``Event Trigger Conditions'' above hold, node $v_j$ chooses an out-neighbor $v_l \in \mathcal{N}_j^+$ according to the order $P_{lj}$ (in a round-robin fashion) and transmits $z_j[k+1]$ and $y_j[k+1]$. 
Then, since it transmitted its stored mass, it sets $y_j[k+1] = 0$, $z_j[k+1] = 0$. 
Regardless of whether it transmitted or not, node $v_j$ sets $k$ to $k+1$ and the iteration repeats (it goes back to Step~1). 

This event-based quantized mass transfer process is summarized as Algorithm~\ref{algorithm1}. 
Note that the ``Event Trigger Conditions'' effectively imply that no transmission is performed if $z_j[k]=0$.

\noindent
\vspace{-0.3cm}    
\begin{varalgorithm}{2}
\caption{Deterministic Quantized Average Consensus}
\textbf{Input} 
\\ 1) A strongly connected digraph $\mathcal{G}_d = (\mathcal{V}, \mathcal{E})$ with $n=|\mathcal{V}|$ nodes and $m=|\mathcal{E}|$ edges. 
\\ 2) For every $v_j$ we have $y_j[0] \in \mathbb{Z}$. 
\\
\textbf{Initialization} 
\\ Every node $v_j \in \mathcal{V}$ does the following:
\\ 1) It assigns to each of its outgoing edges $v_l \in \mathcal{N}^+_j$ a \textit{unique order} $P_{lj}$ in the set $\{0,1,..., \mathcal{D}_j^+ -1\}$.
\\ 2) It sets $z_j[0] = 1$, $z^s_j[0] = 1$ and $y^s_j[0] = y_j[0]$ (which means that $q^s_j[0] = y_j[0] / 1$). 
\\ 3) It chooses an out-neighbor $v_l \in \mathcal{N}_j^+$ according to the predetermined order $P_{lj}$ (initially, it chooses $v_l \in \mathcal{N}_j^+$ such that $P_{lj}=0$) and transmits $z_j[0]$ and $y_j[0]$ to this out-neighbor. Then, it sets $y_j[0] = 0$ and $z_j[0] = 0$.
\\
\textbf{Iteration}
\\ For $k=0,1,2,\dots$, each node $v_j \in \mathcal{V}$ does the following:
\\ 1) It receives $y_i[k]$ and $z_i[k]$ from its in-neighbors $v_i \in \mathcal{N}_j^-$ and sets 
$$
y_j[k+1] = y_j[k] + \sum_{v_i \in \mathcal{N}_j^-} w_{ji}[k]y_i[k] ,
$$
and 
$$
z_j[k+1] = z_j[k] + \sum_{v_i \in \mathcal{N}_j^-} w_{ji}[k]z_i[k] ,
$$
where  $w_{ji}[k] = 0$ if no message is received (otherwise $w_{ji}[k] = 1$). 
\\ 2) \underline{Event Trigger Conditions:} If one of the following two conditions hold, node $v_j$ performs Steps $2a$ and $2b$ below (otherwise it skips Steps $2a$ and $2b$).
\\ Condition~$1$: $z_j[k+1] > z^s_j[k]$.
\\ Condition~$2$: $z_j[k+1] = z^s_j[k]$ and $y_j[k+1] \geq y^s_j[k]$.
\\ 2a) It sets $z^s_j[k+1] = z_j[k+1]$ and $y^s_j[k+1] = y_j[k+1]$ which implies that 
$$
q^s_j[k+1] = \frac{y^s_j[k+1]}{z^s_j[k+1]} .
$$
2b) It chooses an out-neighbor $v_l \in \mathcal{N}_j^+$ according to the order $P_{lj}$ (in a round-robin fashion) and transmits $z_j[k+1]$ and $y_j[k+1]$. Then it sets $y_j[k+1] = 0$ and $z_j[k+1] = 0$.
\\ 3) It repeats (increases $k$ to $k + 1$ and goes back to Step~1).
\label{algorithm1}
\end{varalgorithm}

%The intuition behind the proposed algorithm is that each node $v_j$ receives the mass variables $y_i[k]$ and $z_i[k]$ from its in-neighbors $v_i \in \mathcal{N}_j^-$ and sums them along with its stored mass variables ($y_j[k]$ and $z_j[k]$). 
%Then, it compares the total received mass against its state variables $y^s_j[k]$ and $z^s_j[k]$.
%If the event-trigger conditions do not hold, it performs no actions. 
%If the event-trigger conditions hold, it sets the state variables to be equal to the stored mass, chooses an out-neighbor $v_l \in \mathcal{N}_j^+$ according to the predetermined priority, in a round-robin fashion, and it transmits the stored mass variables.  
%Then, it sets its stored mass variables equal to zero and repeats the above procedure. 

% \subsection{Finite Time Termination of Deterministic Averaging Algorithm}

We now analyze the functionality of the distributed algorithm and prove that it allows all agents to reach quantized average consensus after a finite number of steps. Depending on the graph structure and the initial mass variables of each node, we have the following two possible scenarios:
\begin{enumerate}
\item[A.] Full Mass Summation (i.e., there exists $k'_0 \in \mathbb{Z}_+$ where we have $y_j[k'_0] =\sum_{l=1}^{n}{y_l[0]}  \ \ \text{and} \ \ z_j[k'_0] = n$, for some node $v_j \in \mathcal{V}$, and $y_i[k'_0] = 0  \ \ \text{and} \ \ z_i[k'_0] = 0$, for each $ v_i \in \mathcal{V} - \{ v_j \}$). 
In this scenario (\ref{alpha_z_y}) and (\ref{alpha_q}) hold for each node $v_j$ for the case where $\alpha = 1$. 
\item[B.] Partial Mass Summation (i.e., there exists $k'_0 \in \mathbb{Z}_+$ so that for every $k \geq k'_0$ there exists a set $\mathcal{V}^p[k] \subseteq \mathcal{V}$ in which we have $y_j[k] = y_i[k]$ and $z_j[k] = z_i[k]$, $\forall v_j, v_i \in \mathcal{V}^p[k]$ and $y_l[k] = 0  \ \ \text{and} \ \ z_l[k] = 0$, for each $ v_l \in \mathcal{V} - \mathcal{V}^p[k]$). 
In this scenario, (\ref{alpha_z_y}) and (\ref{alpha_q}) hold for each node $v_j$ for the case where $\alpha = | \mathcal{V}^p[k] |$.
\end{enumerate}

%\noindent
%For the first scenario (full mass summation) we have that, after a finite number of steps, there exists time step $k_0 \in \mathbb{Z}_+$ where one node (say $v_j$) will have mass variables equal to the summation of all the initial mass variables (i.e., $y_j[k_0] =\sum_{l=1}^{n}{y_l[0]}  \ \ \text{and} \ \ z_j[k_0] = n$). 
%For the second scenario (partial mass summation) we have that, there exists time step $k_0 \in \mathbb{Z}_+$ so that for every $k \geq k_0$, there exists a set $\mathcal{V}^p[k]$ (where $\mathcal{V}^p[k] \subseteq \mathcal{V}$) in which nodes $v_j$ and $v_i$ have equal mass variables (i.e., $y_j[k_0] = y_i[k_0]$ and $z_j[k_0] = z_i[k_0]$, $\forall v_j, v_i \in \mathcal{V}^p[k]$) and every node $v_l$ not in the set $\mathcal{V}^p[k]$ has mass variables equal to zero. 

An example regarding the scenario of ``Partial Mass Summation'' is given below. 

\begin{example}

Consider the strongly connected digraph $\mathcal{G}_d = (\mathcal{V}, \mathcal{E})$ shown in Fig.~\ref{partial_example}, with $\mathcal{V} = \{ v_1, v_2, v_3, v_4 \}$ and $\mathcal{E} = \{ m_{21}, m_{32}, m_{43}, m_{14} \}$, where each node has an initial quantized value $y_1[0] = 9$, $y_2[0] = 3$, $y_3[0] = 9$, and $y_4[0] = 3$, respectively. 
We have that the average of the initial values of the nodes, is equal to $q = \frac{24}{4} = 6$. 

\begin{figure}[h]
\begin{center}
\includegraphics[width=0.25\columnwidth]{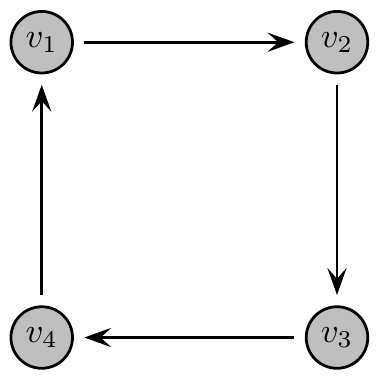}
\caption{Example of digraph for partial mass summation during the operation of Algorithm~\ref{algorithm1}.}
\label{partial_example}
\end{center}
\end{figure}

At time step $k=0$ the initial mass and state variables for nodes $v_1, v_2, v_3, v_4$ are shown in Table~\ref{table_det}. 

\begin{center}
\captionof{table}{Initial Mass and State Variables for Fig.~\ref{partial_example}}
\label{table_det}
{\small 
\begin{tabular}{|c||c|c|c|c|c|}
\hline
Node &\multicolumn{5}{c|}{Mass and State Variables for $k=0$}\\
 &$y_j[0]$&$z_j[0]$&$y^s_j[0]$&$z^s_j[0]$&$q^s_j[0]$\\
\cline{1-6}
 &  &  &  &  & \\
$v_1$ & 9 & 1 & 9 & 1 & 9 / 1\\
$v_2$ & 3 & 1 & 3 & 1 & 3 / 1\\
$v_3$ & 9 & 1 & 9 & 1 & 9 / 1\\
$v_4$ & 3 & 1 & 3 & 1 & 3 / 1\\
\hline
\end{tabular}
}
\end{center}
\vspace{0.2cm}

\noindent
Then, during time step $k=0$, every node $v_j$ will transmit its mass variables $y_j[0]$ and $z_j[0]$ (since the ``Event Trigger Conditions'' hold for every node). 
The mass and state variables of every node at $k=1$ are shown in Table~\ref{table_det_1}.

\noindent
It is important to notice here that, at time step $k=1$, nodes $v_1$ and $v_3$ have mass variables equal to $y_1[1] = 3$, $z_1[1] = 1$ and $y_3[1] = 3$, $z_3[1] = 1$ but the corresponding state variables are equal to $y^s_1[1] = 9$, $z^s_1[1] = 1$ and $y^s_3[1] = 9$, $z^s_3[1] = 1$. 
This means that at time step $k=1$, the ``Event Trigger Conditions'' do \textit{not} hold for nodes $v_1$ and $v_3$; thus, these nodes will not transmit their mass variables (i.e., they will not execute Steps~$2a$ and $2b$ of Algorithm~\ref{algorithm1}). 
The mass and state variables of every node at $k=2$ are shown in Table~\ref{table_det_2}.

During time step $k=2$ we can see that the ``Event Trigger Conditions'' hold for nodes $v_1$ and $v_3$ which means that they will transmit their mass variables towards nodes $v_2$ and $v_4$ respectively. 
The mass and state variables of every node for $k=3$ are shown in Table~\ref{table_det_3}.

\begin{center}
\captionof{table}{Mass and State Variables for Fig.~\ref{partial_example} for $k=1$}
\label{table_det_1}
{\small 
\begin{tabular}{|c||c|c|c|c|c|}
\hline
Node &\multicolumn{5}{c|}{Mass and State Variables for $k=1$}\\
 &$y_j[1]$&$z_j[1]$&$y^s_j[1]$&$z^s_j[1]$&$q^s_j[1]$\\
\cline{1-6}
 &  &  &  &  & \\
$v_1$ & 3 & 1 & 9 & 1 & 9 / 1\\
$v_2$ & 9 & 1 & 9 & 1 & 9 / 1\\
$v_3$ & 3 & 1 & 9 & 1 & 9 / 1\\
$v_4$ & 9 & 1 & 9 & 1 & 9 / 1\\
\hline
\end{tabular}
}
\end{center}
\vspace{0.4cm}

\begin{center}
\captionof{table}{Mass and State Variables for Fig.~\ref{partial_example} for $k=2$} 
\label{table_det_2}
{\small 
\begin{tabular}{|c||c|c|c|c|c|}
\hline
Node &\multicolumn{5}{c|}{Mass and State Variables for $k=2$}\\
 &$y_j[2]$&$z_j[2]$&$y^s_j[2]$&$z^s_j[2]$&$q^s_j[2]$\\
\cline{1-6}
 &  &  &  &  & \\
$v_1$ & 12 & 2 & 12 & 2 & 12 / 2\\
$v_2$ & 0 & 0 & 9 & 1 & 9 / 1\\
$v_3$ & 12 & 2 & 12 & 2 & 12 / 2\\
$v_4$ & 0 & 0 & 9 & 1 & 9 / 1\\
\hline
\end{tabular}
}
\end{center}
\vspace{0.2cm}

\begin{center}
\captionof{table}{Mass and State Variables for Fig.~\ref{partial_example} for $k=3$}
\label{table_det_3}
{\small 
\begin{tabular}{|c||c|c|c|c|c|}
\hline
Node &\multicolumn{5}{c|}{Mass and State Variables for $k=3$}\\
 &$y_j[3]$&$z_j[3]$&$y^s_j[3]$&$z^s_j[3]$&$q^s_j[3]$\\
\cline{1-6}
 &  &  &  &  & \\
$v_1$ & 0 & 0 & 12 & 2 & 12 / 2\\
$v_2$ & 12 & 2 & 12 & 2 & 12 / 2\\
$v_3$ & 0 & 0 & 12 & 2 & 12 / 2\\
$v_4$ & 12 & 2 & 12 & 2 & 12 / 2\\
\hline
\end{tabular}
}
\end{center}
\vspace{0.2cm}

%From Table~\ref{table_det_3}, it is important to notice that, for $k=3$, we have $y_2[3] = y_4[3] = 12$, $z_2[3] = z_4[3] = 2$ and $y_1[3] = y_3[3] = 0$, $z_1[3] = z_3[3] = 0$. 
Following the algorithm operation we have that, for $k=3$, the ``Event Trigger Conditions'' hold for nodes $v_2$ and $v_4$, which means that they will transmit their masses to nodes $v_1$ and $v_3$ respectively. 
As a result we have, for $k=4$, that the mass variables for nodes $v_1$ and $v_3$ are $y_1[4] = y_4[3] = 12$, $z_1[4] = z_4[3] = 2$ and $y_3[4] = y_2[3] = 12$, $z_3[4] = z_2[3] = 2$ respectively. 
Then, during time step $k=4$, we have that the ``Event Trigger Conditions'' hold for nodes $v_1$ and $v_3$ which means that they will transmit their mass variables to nodes $v_1$ and $v_3$. 
We can easily notice that, during the execution of Algorithm~\ref{algorithm1} for $k \geq 3$, we have $\mathcal{V}^p[k] = \mathcal{V}^p[k + 2]$ (where $\mathcal{V}^p[3] = \{ v_2, v_4 \}$ and $\mathcal{V}^p[4] = \{ v_1, v_3 \}$), which means that the exchange of mass variables between the nodes will follow a \textit{periodic} behavior and the mass variables will never ``merge'' in one node (i.e., $\nexists k'_0$ for which $y_j[k'_0] =\sum_{l=1}^{4}{y_l[0]}  \ \ \text{and} \ \ z_j[k'_0] = 4$, for some node $v_j \in \mathcal{V}$, and $y_i[k'_0] = 0  \ \ \text{and} \ \ z_i[k'_0] = 0$, for each $ v_i \in \mathcal{V} - \{ v_j \}$).

Nevertheless, after a finite number of steps, every node $v_j$ obtains a quantized fraction $q_j^s$ which is equal to the average $q$ of the initial values of the nodes. 
From Table~\ref{table_det_3}, we can see that for $k \geq 3$ it holds that 
$$
q_j^s[k] = q = \frac{24 / \alpha}{4 / \alpha} ,
$$
for every $v_j \in \mathcal{V}$, for $\alpha = | \mathcal{V}^p[k] | = 2$. 
\hspace*{\fill} $\square$
\end{example}

\begin{remark}
Note that the periodic behavior in the above graph is not only a function of the graph structure but also of the initial conditions. 
Also note that, in general, the priorities will also play a role because they determine the order in which nodes transmit to their out-neighbors (in the example, priorities do not come into play because each node has exactly one out-neighbor).
\end{remark}

\subsection{Deterministic Convergence Analysis}

For the development of the necessary results regarding the operation of Algorithm~\ref{algorithm1} let us consider the following setup. 

{\it Setup:} Consider a strongly connected digraph $\mathcal{G}_d = (\mathcal{V}, \mathcal{E})$ with $n=|\mathcal{V}|$ nodes and $m=|\mathcal{E}|$ edges. 
During the execution of Algorithm~\ref{algorithm1}, at time step $k_0$, there is at least one node $v_{j'} \in \mathcal{V}$, for which 
\begin{equation}\label{great_z_prop1_det}
z_{j'}[k_0] \geq z_i[k_0], \ \forall v_i \in \mathcal{V}.
\end{equation}
Then, among the nodes $v_{j'}$ for which (\ref{great_z_prop1_det}) holds, there is at least one node $v_j$ for which 
\begin{equation}\label{great_z_prop2_det}
y_j[k_0] \geq y_{j'}[k_0] , \ v_j, v_{j'} \in \{ v_i \in \mathcal{V} \ | \ (\ref{great_z_prop1_det}) \ \text{holds} \}.
\end{equation}
For notational convenience we will call the mass variables of node $v_j$ for which (\ref{great_z_prop1_det}) and (\ref{great_z_prop2_det}) hold as the ``leading mass'' (or ``leading masses'').

Before showing that Algorithm~\ref{algorithm1} allows each node to reach quantized average consensus after a finite number of steps, we present the following lemma, which is helpful in the development of our results.

\begin{lemma}\label{first_lemma}
Under the above {\it Setup}, the ``leading mass'' or ``leading masses'' at time step $k$, (which may be held by different nodes at different time steps) will always fulfill the ``Event Trigger Conditions'' (Step~$2$ of Algorithm~\ref{algorithm1}). 
This means that the mass variables of node $v_j$ for which (\ref{great_z_prop1_det}) and (\ref{great_z_prop2_det}) hold at time step $k_0$ will be transmitted (at time step $k_0$) by $v_j$ to an out-neighbor $v_l \in \mathcal{N}_j^+$. 
\end{lemma}

\begin{proof}
Let us suppose that, at time step $k_0$, (\ref{great_z_prop1_det}) and (\ref{great_z_prop2_det}) hold for the mass variables of node $v_j$ (i.e., it is the ``leading mass''). 
We will show that, during time step $k_0$, the state variables $z_i^s[k_0]$ and $y_i^s[k_0]$ of every node $v_i \in \mathcal{V}$, satisfy one of the following:
\begin{enumerate}
\item $z_i^s[k_0] < z_j[k_0]$ or, 
\item $z_i^s[k_0] = z_j[k_0]$ and $y_i^s[k_0] \leq y_j[k_0]$, 
\end{enumerate}
which means that the mass variables $z_j[k_0]$ and $y_j[k_0]$ of node $v_j \in \mathcal{V}$ will fulfill the ``Event Trigger Conditions'' at time step $k_0$ in Step~$2$ of Algorithm~\ref{algorithm1} at time step $k_0$. 

By contradiction let us suppose that, at time step $k_0$, there exists a node $v_i \in \mathcal{V}$ for which one of the following holds:
\begin{enumerate}
\item $z_i^s[k_0] > z_j[k_0]$ or, 
\item $z_i^s[k_0] = z_j[k_0]$ and $y_i^s[k_0] > y_j[k_0]$, 
\end{enumerate}
while the mass variables of node $v_j$ are the ``leading mass''. 
From Step~$2$ of Algorithm~\ref{algorithm1}, we have that, if the ``Event Trigger Conditions'' hold then each node $v_i$ sets its state variables equal to its mass variables. 
This means that at some past time step $k'_0$ ($k'_0 \leq k_0$), there was a node $v_{l'}$ such that $z_{l'}[k'_0] = z_i^s[k'_0]$ and $y_{l'}[k'_0] = y_i^s[k'_0]$; furthermore, since nonzero masses (like the mass held by node $v_{l'}$) can only remain the same or be merged with other masses, we know that there exists a node $v_l \in \mathcal{V}$ for which, at time step $k_0$, we have one of the following:
\begin{enumerate}
\item $z_l[k_0] = z_i^s[k_0]$ and $y_l[k_0] \geq y_i^s[k_0]$ or, 
\item $z_l[k_0] > z_i^s[k_0]$. 
\end{enumerate}
However, this also means that 
\begin{enumerate}
\item $z_l[k_0] = z_j[k_0]$ and $y_l[k_0] > y_j[k_0]$ or, 
\item $z_l[k_0] > z_j[k_0]$, 
\end{enumerate}
which is a contradiction because (\ref{great_z_prop1_det}) and (\ref{great_z_prop2_det}) do {\it not} hold for the mass variables of node $v_j$ (i.e., it is not the ``leading mass''). 

As a result we have that the ``leading mass'' will always fulfill the ``Event Trigger Conditions'' (Step~$2$ of Algorithm~\ref{algorithm1}); thus, it will always be transmitted to an out-neighbor of the node it is held by. 
\end{proof}

\begin{prop}
\label{PROP1_det}
Under the above {\it Setup} we have that the execution of Algorithm~\ref{algorithm1} allows each node $v_j \in \mathcal{V}$ to reach quantized average consensus after a finite number of steps, bounded by $O(nm^2)$. 
\end{prop} 

\begin{proof}
According to Lemma~\ref{first_lemma}, we have that the ``leading mass'' (which may be held by different nodes at different time steps) will always fulfill the ``Event Trigger Conditions'' (Step~$2$ of Algorithm~\ref{algorithm1}). 
This means that the mass variables of node $v_j$ for which (\ref{great_z_prop1_det}) and (\ref{great_z_prop2_det}) hold, at time step $k$, will always be transmitted by $v_j$ to an out-neighbor $v_l \in \mathcal{N}_j^+$ at time step $k$.

Let us assume that, at time step $k_0$, the mass variables of node $v_j$ are the ``leading mass'' and there exists a set of nodes $\mathcal{V}^f [k_0] \subseteq \mathcal{V}$ which is defined as $\mathcal{V}^f [k_0] = \{ v_i \in \mathcal{V} \ | \ z_i[k_0] > 0 \ \text{but} \ (\ref{great_z_prop1_det}) \ \text{or} \ (\ref{great_z_prop2_det}) \ \text{do not hold} \}$ (i.e., it is the set of nodes which have nonzero mass variables at time step $k_0$ but they are not ``leading masses'').
Note here that if the ``leading mass'' reaches a node simultaneously with some other (leading or otherwise mass) then it gets ``merged'', i.e., the receiving node ``merges'' the mass variables it receives, by summing their numerators and their denominators, creating a set of mass variables with a greater denominator (if necessary, the receiving node updates its state variables to be equal to these merged variables and then propagates them to an out-neighbor). 
Furthermore, we will say that a leading mass, at time step $k_0$, gets ``obstructed'' if it reaches, at time step $k_0 + 1$, a node whose state variables are greater than the mass variables (i.e., either the denominator of the node's state variables is greater than the denominator of the mass variables, or, if the denominators are equal, the numerator of the state variables is greater than the numerator of the mass variables). 
Note that if the ``leading mass'' gets ``obstructed'' then its no longer the ``leading mass'' (there is a new leading mass held by some node in the network).

Suppose that the leading mass at time step $k_0$ is held by node $v_j$ and is given by $y_j[k_0]$ and $z_j[k_0]$. 
If this leading mass does not get merged or obstructed, during the execution of Algorithm~\ref{algorithm1}, it will reach every node $v_j \in \mathcal{V}$ in at most $m^2$ steps, where $m = | \mathcal{E} |$ is the number of edges of the given digraph $\mathcal{G}$ (this follows from Proposition~$3$ in \cite{2014:RikosHadj}, which actually provides a bound for an unobstructed ``leading mass'' to travel via each edge in the graph and thus necessarily also reach every other node\footnote{In Proposition~3 in \cite{2014:RikosHadj} the authors show that the number of iterations required for a packet, which is transmitted between nodes in a round-robin fashion over a directed topology, to reach every node in the network is bounded by $m^2$, where $m$ is the number of edges in the network.}). 
This means that, after executing Algorithm~\ref{algorithm1} for $m^2$ steps, we have $z_i^s[k_0 + m^2] \geq z_j[k_0]$, for every node $v_i \in \mathcal{V}$ (i.e., after $m^2$ steps every node will have its state variable $z^s$ equal or larger than the ``leading mass''). 
Thus, at time step $k_0 + m^2$, if there is any node $v_i \in \mathcal{V}$ (for which we necessarily have $z_i^s[k_0 + m^2] \geq z_j[k_0]$) that belongs in $\mathcal{V}^f [k_0 + m^2]$, this node will perform no transmission (i.e., its mass variables will ``get obstructed'') unless the event triggered conditions hold.

Starting at time step $k_0$, during the execution of Algorithm~\ref{algorithm1}, we examine what happens in the next $m^2$ time steps by considering the following scenarios:

\noindent
A) If in the next $m^2$ time steps the ``leading mass'' gets ``merged'' then we have that at least one ``merging'' occurred (i.e., two nonegative mass variables simultaneously reached a common node in this time window of length $m^2$).

\noindent
B) If the ``leading mass'' at time step $k_0$ gets ``obstructed'', it means that it reached a node (say $v_i$) which performed no transmissions.
For this to happen, then node $v_i$ (or some other node) has ``merged'', at some earlier point, a set of mass variables so that the resulting mass is larger than the ``leading mass'' (i.e., either the denominator of the node's state variables is greater than the denominator of the ``leading mass'', or, if the denominators are equal, the numerator of the other mass variables is greater than the numerator of the ``leading mass''). 
However, this means that there was at least one ``merging'' in these $m^2$ steps.

\noindent
C) If the ``leading mass'' (or ``leading masses'') does not get ``obstructed'' or ``merged'' during the next $m^2$ time steps, we have $z_i^s[k_0 + m^2] \geq z_j[k_0]$, for every node $v_i \in \mathcal{V}$ (i.e., after $m^2$ steps every node will have its state variable $z^s$ equal or larger than the ``leading mass'').
For this scenario we have the two cases below:

\noindent
i) There is at least one node $v_i$ which has nonzero mass variable $z_i[k_0 + m^2]$, and for which the ``Event Trigger Conditions'' do not hold. 
This node will perform no transmission and, since the ``leading mass'' will visit every node $v_j \in \mathcal{V}$ in the subsequent $m^2$ time steps, we have that the ``leading mass'' will visit also this particular node it will ``merge'' with its nonzero mass variables.

\noindent
ii) All nodes have equal mass variables. 
This means that there exists a set of nodes $\mathcal{V}^l [k_0 + m^2] \subseteq \mathcal{V}$ which is defined as $\mathcal{V}^l [k_0 + m^2] = \{ v_i \in \mathcal{V} \ | \ z_i[k_0 + m^2] > 0 \ \text{and} \ (\ref{great_z_prop1_det}) \ \text{and} \ (\ref{great_z_prop2_det}) \ \text{hold} \}$ (i.e., it is the set of nodes which have nonzero mass variables at time step $k_0 + m^2$ and they are ``leading masses''). 
Note here that for every $v_i, v_l \in \mathcal{V}^l [k_0 + m^2]$ we have $y_i[k_0 + m^2] = y_l[k_0 + m^2]$ and $z_i[k_0 + m^2] = z_l[k_0 + m^2]$. 
Note also that it is possible that these leading masses never merge; however, in such case, the nodes have already reached average consensus since, for every $v_i \in \mathcal{V}^l [k_0]$, we have that 
$$
y_i[k_0] =  \frac{\sum_{l=1}^{n}{y_l[0]}}{\alpha}
$$ 
and 
$$
z_i[k_0] = \dfrac{n}{\alpha}
$$
where $\alpha = |\mathcal{V}^l [k_0]|$ 
which means that 
$$
q_i^s[k_0] = \frac{\sum_{l=1}^{n}{y_l[0]}}{n}. 
$$
Even when there are mergings of these mass variables later on, the average will not change. Since the maximum number of mergings is at most $n - 1$, we have that after $O(nm^2)$ iterations the nodes will be able to calculate the average of their initial values.
\end{proof}

The ``Event Trigger Conditions'' in Algorithm~\ref{algorithm1} allowed the calculation of an upper bound on the number of iterations required for every node $v_j \in \mathcal{V}$ to reach quantized average consensus. 
However, from Lemma~\ref{first_lemma} we have that the ``leading mass'' will always fulfill the ``Event Trigger Conditions'' in Step~$2$ of Algorithm~\ref{algorithm1}.
This means that the ``leading mass'' (or leading masses) will continue being transmitted from each node towards its out-neighbors even though quantized average consensus has already been reached. 
The distributed algorithm presented in the following section aims to address this issue by invoking multiple sets of ``Event Trigger Conditions'', which allow the nodes to cease transmissions once quantized average consensus has been reached.

% ===============================================
%
%
% ALGORITHM DETERMINISTIC EVENT TRIGGERED - MAX VOTING
%
%
% ===============================================
\section{EVENT-TRIGGERED QUANTIZED AVERAGING ALGORITHM WITH MINIMUM MASS SUMMATION}\label{MaxAlgorithm}

Motivated by the need to reduce energy consumption, communication bandwidth, network congestion, and/or processor usage, many researchers have considered the use of event-triggered communication and control \cite{2013:Dimarogonas_Johansson, 2016:nowzari_cortes}. 
In this section we extend Algorithm~\ref{algorithm1} so that, once quantized average consensus is reached, all transmissions are ceased. 
The main idea is to maintain a separate mechanism for broadcasting the state variables, $y^s$ and $z^s$, of each node (as long as they satisfy certain event trigger conditions). 
This way, nodes learn the average but also have a way to decide when (or not) to transmit.

\subsection{Event-Triggered Distributed Algorithm with Minimum Mass Summation}

The details of the proposed distributed algorithm with transmission stopping capabilities can be seen in Algorithm~\ref{algorithm_max} below. 
Here we focus on the event triggering rules that are used to determine when to transmit state variables and/or mass variables.

\noindent
\textbf{Initialization:}
Initialization is as in Algorithm~\ref{algorithm1}, except that after initializing its state variables, each node $v_j$ broadcasts its state variables $z^s_j[0]$ and $y^s_j[0]$ to every out-neighbor $v_l \in \mathcal{N}_j^+$.

\noindent
\textbf{Iteration:} The iteration involves the following steps.

\noindent
\textbf{Step 1.} Each node $v_j$ receives $y^s_i[k]$ and $z^s_i[k]$ from its in-neighbors $v_i \in \mathcal{N}_j^-$ (where $y^s_i[k] = 0$ and $z^s_i[k] = 0$ if no message is received).

\noindent
\textbf{Step 2.} Event-Trigger Conditions~$1$: Each node $v_j$, checks the following conditions for every $v_i \in \mathcal{N}_j^-$:
\begin{enumerate}
\item It checks whether $z^s_i[k]$ is greater than $z^s_j[k]$.
\item If $z^s_i[k]$ is equal to $z^s_j[k]$, it checks whether $y^s_i[k]$ is greater than $y^s_j[k]$.
\end{enumerate}
If one of the above two conditions holds, it sets 
\begin{eqnarray}\label{priorities_max}
z^s_j[k+1] = & \max_{v_i \in \mathcal{N}_j^-} z^s_i[k], & \text{and} \nonumber \\
y^s_j[k+1] = & \max_{v_i \in \{v_{i'} \in \mathcal{N}_j^- | z^s_{i'}[k] = z^s_j[k+1]\}} y^s_i[k],  &  \nonumber \;
\end{eqnarray}
which means $q^s_j[k+1] = \frac{y^s_j[k+1]}{z^s_j[k+1]}$. 
Then it broadcasts $z^s_j[k+1]$ and $y^s_j[k+1]$ to every out-neighbor $v_l \in \mathcal{N}_j^+$.

\noindent
\textbf{Step 3.} Event-Trigger Conditions~$2$: Each node $v_j$ checks the following conditions:
\begin{enumerate}
\item It checks whether $z_j[k]$ is lower than $z^s_j[k+1]$.
\item If $z_j[k]$ is equal to $z^s_j[k+1]$, it checks whether $y_j[k]$ is lower than $y^s_j[k+1]$.
\end{enumerate}
If one of the above two conditions holds, it chooses an out-neighbor $v_l \in \mathcal{N}_j^+$ according to the order $P_{lj}$ (in a round-robin fashion) and transmits $y_j[k]$ and $z_j[k]$. 
Note that no transmission is necessary if $z_j[k]=0$ (which means that $y_j[0] = 0$). 
Then, since it transmitted its stored mass, it sets $y_j[k] = 0$, $z_j[k] = 0$.

\noindent
\textbf{Step 4.} Each node $v_j$ receives messages $y_i[k]$ and $z_i[k]$ from its in-neighbors $v_i \in \mathcal{N}_j^-$ and sums them along with its stored messages $y_j[k]$ and $z_j[k]$ to obtain
$$
y_j[k+1] = y_j[k] + \sum_{v_i \in \mathcal{N}_j^-} w_{ji}[k]y_i[k] ,
$$
and 
$$
z_j[k+1] = z_j[k] + \sum_{v_i \in \mathcal{N}_j^-} w_{ji}[k]z_i[k] ,
$$
where $w_{ji}[k] = 0$ if no message is received from in-neighbor $v_i \in \mathcal{N}_j^-$; otherwise $w_{ji}[k] = 1$.

\noindent
\textbf{Step 5.} Event-Trigger Conditions~$3$: Each node $v_j$ checks the following conditions:
\begin{enumerate}
\item It checks whether $z_j[k+1]$ is greater than $z^s_j[k+1]$.
\item If $z_j[k+1]$ is equal to $z^s_j[k+1]$, it checks whether $y_j[k+1]$ is greater than $y^s_j[k+1]$.
\end{enumerate}
If one of the above two conditions holds, it sets $z^s_j[k+1] = z_j[k+1]$ and $y^s_j[k+1] = y_j[k+1]$. Then it broadcasts $z^s_j[k+1]$ and $y^s_j[k+1]$ to every out-neighbor $v_l \in \mathcal{N}_j^+$.

\noindent
Finally, it sets $q^s_j[k+1] = \frac{y^s_j[k+1]}{z^s_j[k+1]}$. Then, $k$ is set to $k+1$ and the iteration repeats (it goes back to Step~1 of the iterative process).

\begin{remark}
Notice here that each node $v_j$, during time step $k$, performs two types of transmission, towards its out-neighbors $v_l \in \mathcal{N}_j^+$: either via broadcasting (to all of its out-neighbors) of its state variables $y^s_j[k]$ and $z^s_j[k]$ (if ``Event Trigger Conditions~$1$ and $3$'' hold) or via transmission of its mass variables $y_j[k]$ and $z_j[k]$ to a single out-neighbor, chosen according to the predetermined order $P_{lj}$ (if ``Event Trigger Conditions~$2$'' hold). 
The event trigger conditions effectively imply that no transmission is performed if no set of conditions holds in Steps~$2$, $3$ and $5$. 
\end{remark}

\noindent
\vspace{-0.5cm}    
\begin{varalgorithm}{3}
\caption{Deterministic Quantized Average Consensus with Minimum Mass Summation}
\textbf{Input} 
\\ 1) A strongly connected digraph $\mathcal{G}_d = (\mathcal{V}, \mathcal{E})$ with $n=|\mathcal{V}|$ nodes and $m=|\mathcal{E}|$ edges. 
\\ 2) For every $v_j$ we have $y_j[0] \in \mathbb{Z}$. 
\\
\textbf{Initialization} 
\\ Every node $v_j \in \mathcal{V}$ does the following:
\\ 1) It assigns to each of its outgoing edges $v_l \in \mathcal{N}^+_j$ a \textit{unique order} $P_{lj}$ in the set $\{0,1,..., \mathcal{D}_j^+ -1\}$.
\\ 2) It sets $z_j[0] = 1$, $z^s_j[0] = 1$ and $y^s_j[0] = y_j[0]$ (which means that $q^s_j[0] = y_j[0] / 1$). 
%\\ 3) Chooses an out-neighbor $v_l \in \mathcal{N}_j^+$ according to the predetermined order $P_{lj}$ (i.e., it chooses $v_l \in \mathcal{N}_j^+$ such that $P_{lj}=0$) and transmits $z_j[0]$ and $y_j[0]$ to this out-neighbor. Then, it sets $y_j[0] = 0$ and $z_j[0] = 0$.
\\ 3) It broadcasts its state variables $z^s_j[0]$ and $y^s_j[0]$ to every out-neighbor $v_l \in \mathcal{N}_j^+$.
\\
\textbf{Iteration}
\\ For $k=0,1,2,\dots$, each node $v_j \in \mathcal{V}$ does the following: 
\\ 1) It receives $y^s_i[k]$ and $z^s_i[k]$ from its in-neighbors $v_i \in \mathcal{N}_j^-$ 
(where $y^s_i[k] = 0$ and $z^s_i[k] = 0$ if no message is received). 
\\ 2) \underline{Event Trigger Conditions 1:} Node $v_j$ checks if there exists $v_i \in \mathcal{N}_j^-$ for which one of the following two conditions hold:
\\ Condition~$(i)$: $z^s_i[k] > z^s_j[k]$.
\\ Condition~$(ii)$: $z^s_i[k] = z^s_j[k]$ and $y^s_i[k] > y^s_j[k]$.
\\ If one of the two conditions above holds, node $v_j$ sets 
$$ 
z^s_j[k+1] = \max_{v_i \in \mathcal{N}_j^-} z^s_i[k] , \ \ \text{and}
$$ 
$$ 
y^s_j[k+1] = \max_{v_i \in \{v_{i'} \in \mathcal{N}_j^- | z^s_{i'}[k] = z^s_j[k+1]\}} y^s_i[k] ,
$$ 
which means $q^s_j[k+1] = \frac{y^s_j[k+1]}{z^s_j[k+1]}$. 
%Then, it broadcasts $z^s_j[k+1]$ and $y^s_j[k+1]$ to every out-neighbor $v_l \in \mathcal{N}_j^+$.
\\ 3) \underline{Event Trigger Conditions 2:} Node $v_j$ checks if one of the following two conditions hold:
\\ Condition~$(i)$: $0 < z_j[k] < z^s_j[k+1]$.
\\ Condition~$(ii)$: $z_j[k] = z^s_j[k+1]$ and $y_j[k] < y^s_j[k+1]$. 
If one of the two conditions above holds, then node $v_j$ chooses an out-neighbor $v_l \in \mathcal{N}_j^+$ according to the order $P_{lj}$ (in a round-robin fashion) and transmits $y_j[k]$ and $z_j[k]$. Then it sets $y_j[k] = 0$ and $z_j[k] = 0$.
\\ 4) It receives $y_i[k]$ and $z_i[k]$ from its in-neighbors $v_i \in \mathcal{N}_j^-$ and sets 
$$
y_j[k+1] = y_j[k] + \sum_{v_i \in \mathcal{N}_j^-} w_{ji}[k]y_i[k] ,
$$
and 
$$
z_j[k+1] = z_j[k] + \sum_{v_i \in \mathcal{N}_j^-} w_{ji}[k]z_i[k] ,
$$
where  $w_{ji}[k] = 0$ if no message is received (otherwise $w_{ji}[k] = 1$).
\\ 5) \underline{Event Trigger Conditions 3:} Node $v_j$ checks if one of the following two conditions hold:
\\ Condition~$(i)$: $z_j[k+1] > z^s_j[k+1]$.
\\ Condition~$(ii)$: $z_j[k+1] = z^s_j[k+1]$ and $y_j[k+1] > y^s_j[k+1]$.
If one of the two conditions above holds, then node $v_j$ sets $z^s_j[k+1] = z_j[k+1]$ and $y^s_j[k+1] = y_j[k+1]$ which means $q^s_j[k+1] = \frac{y^s_j[k+1]}{z^s_j[k+1]}$. 
\\ 6) If either ``Event Trigger Conditions~$2$'' or ``Event Trigger Conditions~$3$'' hold, it broadcasts $z^s_j[k+1]$ and $y^s_j[k+1]$ to every out-neighbor $v_l \in \mathcal{N}_j^+$.
\\ 7) It repeats (increases $k$ to $k + 1$ and goes back to Step~1).
\label{algorithm_max}
\end{varalgorithm}

\vspace{.1cm}

We now analyze the functionality of the distributed algorithm and prove that it allows all agents to reach quantized average consensus after a finite number of steps. 
Furthermore, we will also show that once quantized average consensus is reached transmissions are ceased from each agent. 
Depending on the graph structure and the initial mass variables of each node, we have the following two possible scenarios: ``Full Mass Summation'' or ``Partial Mass Summation'' (as presented in Section~\ref{DetAlgorithm}). 
An example regarding the scenario of ``Partial Mass Summation'' is given below.

\begin{example}

Consider a strongly connected digraph $\mathcal{G}_d = (\mathcal{V}, \mathcal{E})$, shown in Fig.~\ref{max_example}, with $\mathcal{V} = \{ v_1, v_2, v_3, v_4 \}$ and $\mathcal{E} = \{ m_{31}, m_{41}, m_{12}, m_{13}, m_{43}, m_{24} \}$ where each node has an initial quantized value $y_1[0] = 2$, $y_2[0] = 4$, $y_3[0] = 7$ and $y_4[0] = 9$ respectively. 
The average of the initial values, is equal to $q = \frac{22}{4}$. 

\begin{figure}[h]
\begin{center}
\includegraphics[width=0.25\columnwidth]{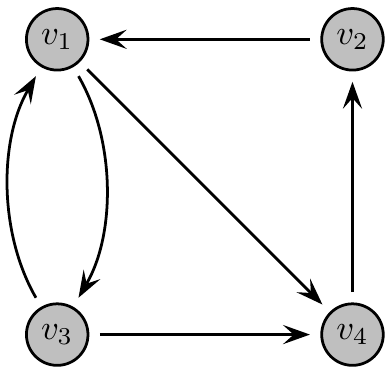}
\caption{Example of digraph for partial mass summation when using Algorithm~\ref{algorithm_max}.}
\label{max_example}
\end{center}
\end{figure}

Each node $v_j \in \mathcal{V}$ follows the Initialization steps ($1-2$) in Algorithm~\ref{algorithm_prob}, assigning to each of its outgoing edges $v_l \in \mathcal{N}^+_j$ a unique order $P_{lj}$ in the set $\{0,1,..., \mathcal{D}_j^+ -1\}$. 
Assume that the unique orders assigned by each node are the following:
\begin{eqnarray}\label{priorities_max}
v_1: & P_{41} = 0,  & P_{31} = 1, \nonumber \\
v_2: & P_{12} = 0,  &  \nonumber \\ 
v_3: & P_{13} = 0,  & P_{43} = 1,  \nonumber \\ 
v_4: & P_{24} = 0.  &  \nonumber \;
\end{eqnarray}
Furthermore, each node broadcasts its state variables $z^s_j[0]$ and $y^s_j[0]$ to every out-neighbor $v_l \in \mathcal{N}_j^+$.
The initial mass and state variables, at time step $k=0$, for nodes $v_1, v_2, v_3, v_4$ are shown in Table~\ref{table_max1}. 

\begin{center}
\captionof{table}{Initial Mass and State Variables for Fig.~\ref{max_example}}
\label{table_max1}
{\small 
\begin{tabular}{|c||c|c|c|c|c|}
\hline
Node &\multicolumn{5}{c|}{Mass and State Variables for $k=0$}\\
 &$y_j[0]$&$z_j[0]$&$y^s_j[0]$&$z^s_j[0]$&$q^s_j[0]$\\
\cline{1-6}
 &  &  &  &  & \\
$v_1$ & 2 & 1 & 2 & 1 & 2 / 1\\
$v_2$ & 4 & 1 & 4 & 1 & 4 / 1\\
$v_3$ & 7 & 1 & 7 & 1 & 7 / 1\\
$v_4$ & 9 & 1 & 9 & 1 & 9 / 1\\
\hline
\end{tabular}
}
\end{center}
\vspace{0.2cm}

During the operation of Algorithm~\ref{algorithm_max}, at time step $k=0$, each node $v_j$ will receive the state variables $z^s_i[0]$ and $y^s_i[0]$ from every in-neighbor $v_i \in \mathcal{N}_j^-$. 
According to the ``Event Trigger Conditions~$1$'', each node will update its state variables. 
Here, we have that nodes $v_1$ and $v_2$ will update them. 
Furthermore, following ``Event Trigger Conditions~$2$'', nodes $v_1$ and $v_2$ will transmit their mass variables according to their unique predetermined order (nodes $v_3$ and $v_4$ will not transmit their mass variables because ``Event Trigger Conditions~$2$'' do not hold for them). 
Then, nodes $v_1$ and $v_4$ will receive the mass variables from their in-neighbors (i.e., from $v_2$ and $v_1$ respectively) and, following ``Event Trigger Conditions~$3$'', they will update (and broadcast) their state variables. 
The mass and state variables, at time step $k=1$, for nodes $v_1, v_2, v_3, v_4$ are shown in Table~\ref{table_max2}. 

During time step $k=1$, each node $v_j$ will receive the state variables from every in-neighbor and following ``Event Trigger Conditions~$1$'', $v_1$ and $v_2$ will update their state variables. 
Then, following ``Event Trigger Conditions~$2$'', node $v_1$ will transmit its mass variables (according to its unique predetermined order) towards node $v_3$. 
Node $v_3$ will receive the mass variables from node $v_1$ and, following ``Event Trigger Conditions~$3$'', it will update (and broadcast towards its out-neighbors) its state variables. 
The mass and state variables, at time step $k=2$, for nodes $v_1, v_2, v_3, v_4$ are shown in Table~\ref{table_max3}. 

\begin{center}
\captionof{table}{Mass and State Variables for Fig.~\ref{max_example} for $k=1$}
\label{table_max2}
{\small 
\begin{tabular}{|c||c|c|c|c|c|}
\hline
Node &\multicolumn{5}{c|}{Mass and State Variables for $k=1$}\\
 &$y_j[1]$&$z_j[1]$&$y^s_j[1]$&$z^s_j[1]$&$q^s_j[1]$\\
\cline{1-6}
 &  &  &  &  & \\
$v_1$ & 4 & 1 & 7 & 1 & 7 / 1\\
$v_2$ & 0 & 0 & 9 & 1 & 9 / 1\\
$v_3$ & 7 & 1 & 7 & 1 & 7 / 1\\
$v_4$ & 11 & 2 & 11 & 2 & 11 / 2\\
\hline
\end{tabular}
}
\end{center}
\vspace{0.2cm}

\begin{center}
\captionof{table}{Mass and State Variables for Fig.~\ref{max_example} for $k=2$}
\label{table_max3}
{\small 
\begin{tabular}{|c||c|c|c|c|c|}
\hline
Node &\multicolumn{5}{c|}{Mass and State Variables for $k=2$}\\
 &$y_j[2]$&$z_j[2]$&$y^s_j[2]$&$z^s_j[2]$&$q^s_j[2]$\\
\cline{1-6}
 &  &  &  &  & \\
$v_1$ & 0 & 0 & 9 & 1 & 9 / 1\\
$v_2$ & 0 & 0 & 9 & 1 & 9 / 1\\
$v_3$ & 11 & 2 & 11 & 2 & 11 / 2\\
$v_4$ & 11 & 2 & 11 & 2 & 11 / 2\\
\hline
\end{tabular}
}
\end{center}
\vspace{0.2cm}

In Table~\ref{table_max3} we can see that for the set $\mathcal{V}^p[2] = \{v_3, v_4 \}$ we have $y_3[2] = y_4[2]$ and $z_3[2] = z_4[2]$ while, for the set $ \mathcal{V} - \mathcal{V}^p[2] = \{v_1, v_2 \}$ we have 
$y_1[2] = y_2[2] = 0$ and $z_1[2] = z_2[2] = 0$. 
This means that we have a ``Partial Mass Summation'' scenario. 
In this case, we will see that ``Event Trigger Conditions~$3$'' will not hold again for any node for time steps $k > 2$ (i.e., no node will transmit again its mass variables).

In the next time step ($k=2$), once each node receives the state variables from every in-neighbor, $v_2$ will update and broadcast its state variables (according to the ``Event Trigger Conditions~$1$''). 
Then, since the ``Event Trigger Conditions~$2$'' do not hold for any node, no transmissions of mass variables will be performed; thus, no node will receive any mass variables.
This means that the ``Event Trigger Conditions~$3$'' also do not hold for any node. 
The mass and state variables, at time step $k=3$, for nodes $v_1, v_2, v_3, v_4$ are shown in Table~\ref{table_max4}. 

\begin{center}
\captionof{table}{Mass and State Variables for Fig.~\ref{max_example} for $k=3$}
\label{table_max4}
{\small 
\begin{tabular}{|c||c|c|c|c|c|}
\hline
Node &\multicolumn{5}{c|}{Mass and State Variables for $k=3$}\\
 &$y_j[3]$&$z_j[3]$&$y^s_j[3]$&$z^s_j[3]$&$q^s_j[3]$\\
\cline{1-6}
 &  &  &  &  & \\
$v_1$ & 0 & 0 & 9 & 1 & 9 / 1\\
$v_2$ & 0 & 0 & 11 & 2 & 11 / 2\\
$v_3$ & 11 & 2 & 11 & 2 & 11 / 2\\
$v_4$ & 11 & 2 & 11 & 2 & 11 / 2\\
\hline
\end{tabular}
}
\end{center}
\vspace{0.2cm}

In time step $k=3$, we have that $v_1$ will update and broadcast its state variables (according to the ``Event Trigger Conditions~$1$''). 
Then, since the ``Event Trigger Conditions~$2$'' and ``Event Trigger Conditions~$3$'' do not hold for any node no transmissions will be performed. 
The mass and state variables, at time step $k=4$, for nodes $v_1, v_2, v_3, v_4$ are shown in Table~\ref{table_max5}. 

\begin{center}
\captionof{table}{Mass and State Variables for Fig.~\ref{max_example} for $k=4$}
\label{table_max5}
{\small 
\begin{tabular}{|c||c|c|c|c|c|}
\hline
Node &\multicolumn{5}{c|}{Mass and State Variables for $k=4$}\\
 &$y_j[4]$&$z_j[4]$&$y^s_j[4]$&$z^s_j[4]$&$q^s_j[4]$\\
\cline{1-6}
 &  &  &  &  & \\
$v_1$ & 0 & 0 & 11 & 2 & 11 / 2\\
$v_2$ & 0 & 0 & 11 & 2 & 11 / 2\\
$v_3$ & 11 & 2 & 11 & 2 & 11 / 2\\
$v_4$ & 11 & 2 & 11 & 2 & 11 / 2\\
\hline
\end{tabular}
}
\end{center}
\vspace{0.2cm}

In Table~\ref{table_max5}, we can see that (\ref{alpha_z_y}) and (\ref{alpha_q}) hold for every node for $\alpha = 2$ (i.e., every node has reached quantized average consensus). 
Notice that no set of event trigger conditions holds for any node for time steps $k \geq 4$. 
This means that no node will perform any transmissions of its state or mass variables for time steps $k \geq 4$. \hspace*{\fill} $\square$
\end{example}

\begin{remark}
It is interesting to note here that if, during the Initialization of Algorithm~\ref{algorithm_max}, node $v_1$ sets its priorities as $P_{31} = 0$ and $P_{41} = 1$, then we will notice the scenario of ``Full Mass Summation'' for node $v_4$ (i.e., (\ref{alpha_z_y}) and (\ref{alpha_q}) hold for every node for $\alpha = 1$) instead of the scenario of ``Partial Mass Summation''. 
\end{remark}

\subsection{Deterministic Convergence Analysis}

The setup is identical to the one presented in Section~\ref{MaxAlgorithm}, where the mass variables of node $v_j$ for which (\ref{great_z_prop1_det}) and (\ref{great_z_prop2_det}) hold at time step $k$ are called the ``leading mass''. 
Furthermore, we will call the mass variables of every node $v_i \in \mathcal{V}$, for which $z_i[k] > 0$ and for which neither (\ref{great_z_prop1_det}) nor (\ref{great_z_prop2_det}) hold at time step $k$, as the ``follower mass''. 

\begin{lemma}\label{before_second_lemma}
If, during time step $k_0$ of Algorithm~\ref{algorithm_max}, the mass variables of node $v_j$ fulfil (\ref{great_z_prop1_det}) and (\ref{great_z_prop2_det}), then the state variables of every node $v_i \in \mathcal{V}$ satisfy 
\begin{eqnarray}\label{first_z}
z_i^s[k_0] \leq z_j[k_0] ,
\end{eqnarray}
or 
\begin{eqnarray}\label{first_zy}
z_i^s[k_0] = z_j[k_0] \ \ \text{and} \ \ y_i^s[k_0] \leq y_j[k_0].
\end{eqnarray}
\end{lemma}

\begin{proof}
Let us consider the variable
$$
z^{(m)}[k] = \max_{v_l \in \mathcal{V}} z_l[k] . 
$$
From Iteration Step~$4$ of Algorithm~\ref{algorithm_max} we have that $z^{(m)}[k]$ is non-decreasing (i.e., $z^{(m)}[k+1] \geq z^{(m)}[k]$, for every $k$). 
Furthermore, since the mass variables of node $v_j$ fulfill (\ref{great_z_prop1_det}) and (\ref{great_z_prop2_det}), then, during every time step $k$, it holds that  
$$
z_j[k] = z^{(m)}[k]. 
$$

In addition, for every $k$, during Iteration Steps~$2$ and $5$ of Algorithm~\ref{algorithm_max}, for every node $v_i \in \mathcal{V}$, we have that $z_i^s[k]$ is either less than $z^{(m)}[k]$ (i.e., $z_i^s[k] < z^{(m)}[k]$) or equal to $z^{(m)}[k]$ (i.e., $z_i^s[k] = z^{(m)}[k]$). 
As a result, at time step $k$, the state variables of every node $v_i \in \mathcal{V}$ satisfy 
$$
z_i^s[k] \leq z_j[k] .
$$

Finally, from Iteration Steps~$2$ and $5$, (for every $k$) if, it holds that $z_i^s[k] = z_j[k]$ for some node $v_i$ then we have that either $y_i^s[k] < y_j[k]$ or $y_i^s[k] = y_j[k]$. 
[Note here that if $z_i^s[k] = z_j[k]$ and $y_i^s[k] > y_j[k]$, then the mass variables of $v_j$ do not fulfil (\ref{great_z_prop1_det}) and (\ref{great_z_prop2_det}) which is a contradiction.]
As a result we have that if the mass variables of node $v_j$ fulfil (\ref{great_z_prop1_det}) and (\ref{great_z_prop2_det}), then the state variables of every node $v_i \in \mathcal{V}$ satisfy (\ref{first_z}) or (\ref{first_zy}). 
\end{proof}

\begin{lemma}\label{second_lemma}
If, during time step $k_0$ of Algorithm~\ref{algorithm_max}, the mass variables of {\em each} node $v_j$ with nonzero mass variables fulfill (\ref{great_z_prop1_det}) and (\ref{great_z_prop2_det}), then we have only ``leading masses'' and no ``follower masses''. 
This means that the ``Event Trigger Conditions~$2$'' will never hold again for future time steps $k \geq k_0$. 
As a result, the transmissions that (may) take place will only be via broadcasting (from ``Event Trigger Conditions~$1$ and $3$'') for at most $n-1$ time steps and then they will cease. 
\end{lemma}

\begin{proof}
Let us assume that during time step $k_0$ two (or more) mass variables merge at nodes $v_j$, $v_i$, so that these two nodes simultaneously become ``leading masses'' (more generally, we could have more than two leading masses) and all other nodes have zero mass variables. 
Since the mass variables of nodes $v_j$, $v_i$, during time step $k_0$, become ``leading masses'' then there exists a set $\mathcal{V}^p[k_0] \subseteq \mathcal{V}$ in which we have $y_j[k_0] = y_i[k_0]$ and $z_j[k_0] = z_i[k_0]$, $\forall v_j, v_i \in \mathcal{V}^p[k_0]$ and $y_l[k_0] = 0  \ \ \text{and} \ \ z_l[k_0] = 0$, for each $ v_l \in \mathcal{V} - \mathcal{V}^p[k_0]$.
Once this ``merge'' occurs then we have that for both $v_j$ and $v_i$ the ``Event Trigger Conditions $1$'' and the ``Event Trigger Conditions $2$'' do not hold, but ``Event Trigger Conditions $3$'' do hold. 
This means that $v_j$ and $v_i$ do not transmit their mass variables but rather they broadcast their new state variables to their out-neighbors. 
Then, their out-neighbors, $v_{l_j}$ and $v_{l_i}$ respectively, will update their state variables and broadcast their new state variables towards their out-neighbors. 
The updating and broadcasting of state variables will continue, until all nodes obtain state variables equal to $z^s_j[k_0]$ and $y^s_j[k_0]$. 
Note that during this update and broadcasting of state variables, no node transmits its mass variables. 
After at most $n-1$ steps, all nodes will be aware of the values $z^s_j[k_0]$ and $y^s_j[k_0]$, and at that point all transmissions will seize.
\end{proof}

\begin{prop}
\label{PROP1_max}
The execution of Algorithm~\ref{algorithm_max} allows each node $v_j \in \mathcal{V}$ to reach quantized average consensus after a finite number of steps, bounded by $O(nm^2)$. 
Furthermore, once quantized average consensus is reached, each node stops transmitting towards its out-neighbors. 
\end{prop}

\begin{proof}
Before starting the analysis of Algorithm~\ref{algorithm_max}, it is important to notice that the ``leading mass'' will not fulfill ``Event Trigger Conditions~$2$'' in Step~$3$ of the Iteration of Algorithm~\ref{algorithm_max} which means that the corresponding node (say $v_j$) will not transmit its mass variables to its out-neighbors $v_l \in \mathcal{N}_j^+$ according to its predetermined priority. 
In this proof, we will show that there exists $k_0 \in \mathbb{Z}_+$, where for every $k \geq k_0$, the mass variables of {\em every} node $v_j$, for which $z_j[k] > 0$, fulfill (\ref{great_z_prop1_det}) and (\ref{great_z_prop2_det}) (i.e., for $k \geq k_0$ we have only ``leading masses''). 
Furthermore, from Lemma~\ref{second_lemma}, we have that there exists $k_1 > k_0$, where for every $k \geq k_1$ the state variables of every node $v_j \in \mathcal{V}$ fulfill (\ref{alpha_z_y}) and (\ref{alpha_q}) for $\alpha \in \mathbb{Z}_+$ (i.e., every node has reached quantized consensus) and thus transmissions cease.

During the Initialization steps of Algorithm~\ref{algorithm_max}, we have that each node will broadcast its state variables to every out-neighbor. 
Then, during Iteration Step~$1$, each node will receive and update its state variables, while during Step~$2$ (i.e., ``Event Trigger Conditions~$1$''), it will broadcast towards its out-neighbors the updated values (of the state variables). 
This means that after $n$ iterations (assuming, that no other mass variables ``merged'' during $n$ time steps), the state variables of each node $v_i \in \mathcal{V}$, satisfy
\begin{equation}\label{great_z_prop3_max}
z^s_i[n] = z_{j_1}[0], \ \ \text{and} \ \ y^s_i[n] = y_{j_1}[0] , \nonumber
\end{equation}
where the mass variables of node $v_{j_1}$ are the ``leading mass''. 
As a result we have that, after $n$ iterations, the ``Event Trigger Conditions~$2$'' will hold for every node $v_i \in \mathcal{V} - \{ v_{j_1} \}$, and thus every node (except node $v_{j_1}$ which is the ``leading mass'') will transmit its mass variables toward its out-neighbors according to its unique priority. 
Note here that the number of iterations required for the ``follower mass'' to reach every node $v_i \in \mathcal{V}$ is bounded by $m^2$, where $m = | \mathcal{E} |$ is the number of edges of the given digraph $\mathcal{G}$ (in this case Proposition~3 in \cite{2014:RikosHadj}, provides a bound for the ``follower mass'' to travel via each edge in the graph and thus necessarily also reach every other node).
Let us assume now that, after executing Algorithm~\ref{algorithm_max} for additional $m^2$ steps, we have that the mass variables $z_{i_1}[0]$, $y_{i_1}[0]$ and $z_{i_2}[0]$, $y_{i_2}[0]$ of nodes $v_{i_1}$ and $v_{i_2}$ respectively, meet (and ``merge'') in node $v_{j_2}$, and after this merge they become the ``leading mass''.  
This means that, during time step $n + m^2$, node $v_{j_2}$ will not transmit its mass variables (because ``Event Trigger Conditions~$2$'' do not hold) but it will broadcast its state variables to every out-neighbor (because ``Event Trigger Conditions~$3$'' hold). 
Thus, after additional $n$ iterations, the state variables of each node $v_i \in \mathcal{V}$ satisfy
\begin{equation}\label{great_z_prop4_max_part1}
z^s_i[2n + m^2] = z_{j_2}[n + m^2], \nonumber
\end{equation}
and
\begin{equation}\label{great_z_prop4_max_part2}
y^s_i[2n + m^2] = y_{j_2}[n + m^2] , \nonumber
\end{equation}
where the mass variables of node $v_{j_2}$ are now the ``leading mass''. 
This means that the ``Event Trigger Conditions~$2$'' will hold for every node $v_i \in \mathcal{V} - \{ v_{j_2} \}$, and thus every node (except node $v_{j_2}$ which is now the ``leading mass'') will transmit its mass variables toward its out-neighbors according to its unique priority. 
Note here that also node $v_{j_1}$ will transmit its mass variables toward its out-neighbors (since the state variables of $v_{j_1}$ are equal to the mass variables of the ``leading mass'' $v_{j_2}$ this means that ``Event Trigger Conditions~$2$'' will also hold for $v_{j_1}$). 
Let us assume now that, after executing Algorithm~\ref{algorithm_max} for additional $m^2$ steps, the mass variables $z_{i_3}[0]$, $y_{i_3}[0]$ and $z_{i_4}[0]$, $y_{i_4}[0]$ of nodes $v_{i_3}$ and $v_{i_4}$ respectively, meet (and ``merge'') in node $v_{j_3}$, and after this merge they become the ``leading mass''. 
Again, this means that during time step $2n + 2m^2$, node $v_{j_3}$ will not transmit its mass variables (because ``Event Trigger Conditions~$2$'' do not hold) but it will broadcast its state variables to every out-neighbor (because ``Event Trigger Conditions~$3$'' hold). 
After additional $n$ iterations, the state variables of each node $v_i \in \mathcal{V}$ satisfy
\begin{equation}\label{great_z_prop5_max_part1}
z^s_i[3n + 2m^2] = z_{j_3}[2n + 2m^2], \nonumber
\end{equation}
and 
\begin{equation}\label{great_z_prop5_max_part2}
y^s_i[3n + 2m^2] = y_{j_3}[2n + 2m^2] , \nonumber
\end{equation}
where the mass variables of node $v_{j_3}$ are now the new ``leading mass''. 
By continuing this analysis, we can see that every $n + m^2$ time steps at least two ``follower masses'' ``merge'' and become the ``leading mass''. 
Since, during the Initialization steps of Algorithm~\ref{algorithm_max}, we have $n$ initial mass variables this means that after $(n-1)n + (n-1)m^2$ time steps {\it all} initial mass variables will ``merge'' into one mass (obviously the mass variables in which every initial mass has ``merged'' is the ``leading mass''). 
As a result, at time step $(n-1)n + (n-1)m^2$ we have only ``leading masses'' and no ``follower masses''. 
Thus, from Lemma~\ref{second_lemma}, we have that after additional $n$ time steps every node will have state variables equal to the ``leading mass'' (i.e., $z^s_i[n^2 + (n-1)m^2] = n$ and $y^s_i[n^2 + (n-1)m^2] = \sum_{l=1}^{n}{y_l[0]}$, for every $v_i \in \mathcal{V}$) and then transmissions will be ceased.

Note that so far we considered the scenario where there is only one ``leading mass'' during every time step $k$ and it ``merges'' with only one nonzero mass variable every $n + m^2$ time steps. 
In other scenarios, we can consider multiple ``leading masses'' (i.e., when the nonzero mass variables fulfill (\ref{great_z_prop1_det}) and (\ref{great_z_prop2_det}) for more than one node) which will also speed up the convergence since the ``follower masses'' will ``merge'' more frequently. 
\end{proof}

\begin{remark}
It is important to note here that the operation of Algorithm~\ref{algorithm_max} follows an opposite scenario than the operation of Algorithm~\ref{algorithm1}. 
In Algorithm~\ref{algorithm1}, the ``leading mass'' always fulfills the ``Event Trigger Conditions'', and thus it is always transmitted from each node according to its unique priority. 
However, in Algorithm~\ref{algorithm_max}, the ``leading mass'' does not fulfill ``Event Trigger Conditions~$2$'' which means that it will not be transmitted. 
This means that the operation of Algorithm~\ref{algorithm_max} leverages on the fact that the number of initial masses is finite (and equal to $n$) and thus their summation into one (or multiple) ``leading mass'' (or ``leading masses'') will occur after a finite number of steps. 
Once this happens, transmissions between nodes will continue for a maximum of $n$ time steps (from ``Event Trigger Conditions~$1$'') and then will cease. 
\end{remark}

\section{SIMULATIONS AND COMPARISONS} \label{results}

In this section, we illustrate the behavior of the proposed distributed algorithms for a random graph of size $n = 20$ nodes, a ring-shaped directed graph of size $n = 20$ nodes and a ring-shaped undirected graph of size $n = 20$ nodes. 
We also compare the proposed algorithms against the current state-of-the-art, trying to point out key differences and limitations in each approach.
Specifically, we first illustrate the operation of Algorithms~\ref{algorithm_prob}, \ref{algorithm1} and \ref{algorithm_max} in digraphs of size $n=20$ nodes. 

Figure~\ref{demon20} shows what happens in the case of a randomly, created graph of $20$ nodes in which the average of the initial values is equal to $q = \dfrac{500}{20} = 25$. 
We can see that Algorithm~\ref{algorithm_max} outperforms Algorithms~\ref{algorithm_prob} and \ref{algorithm1}.

\vspace{-0.3cm}

\begin{figure} [ht]
\centering
\includegraphics[width=80mm]{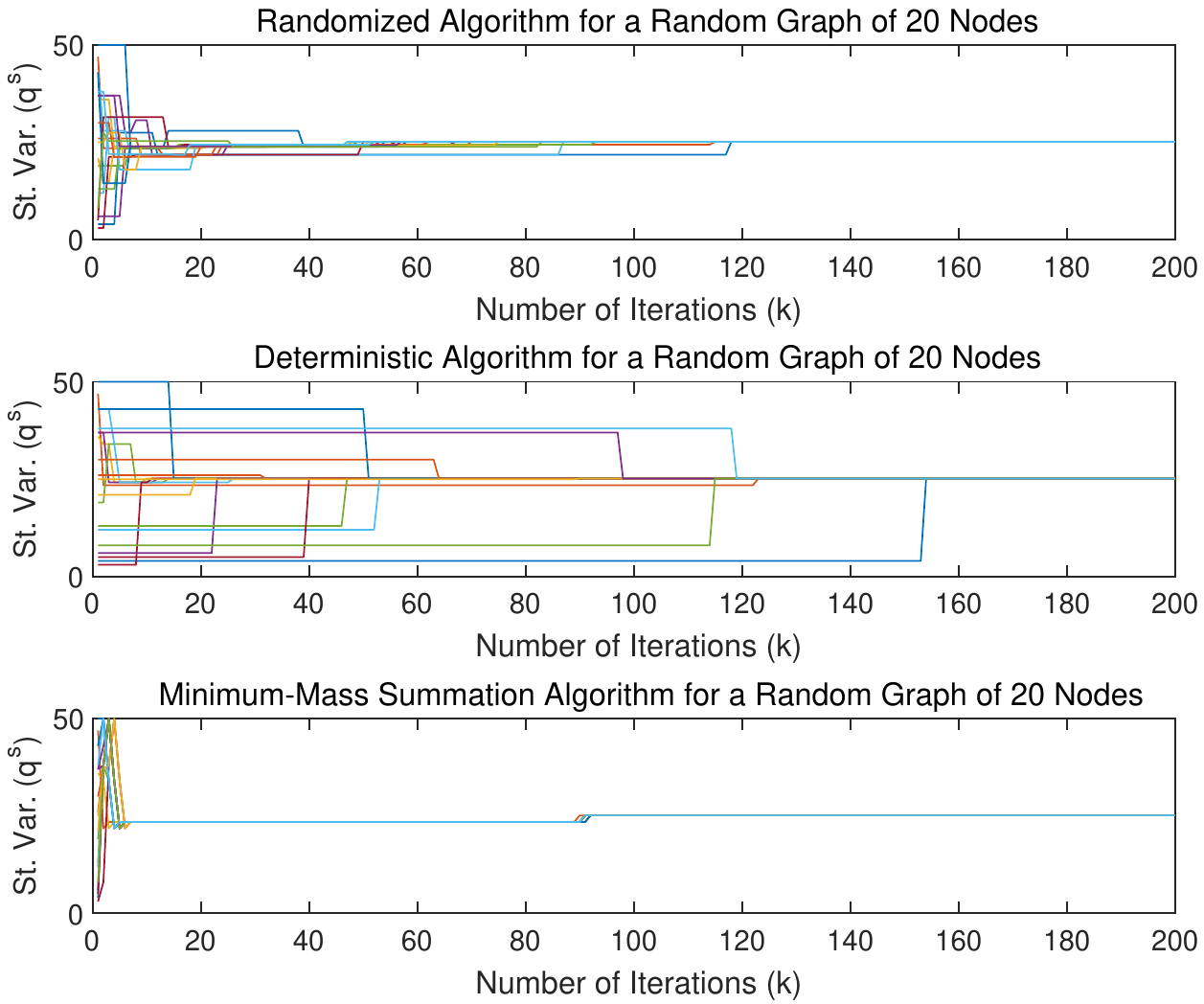}
\caption{Comparison between Algorithm~\ref{algorithm_prob}, Algorithm~\ref{algorithm1} and Algorithm~\ref{algorithm_max} for a random digraph of $20$ nodes: 
Node state variables plotted against the number of iterations for Algorithm~\ref{algorithm_prob} (\emph{top figure}), Algorithm~\ref{algorithm1} (\emph{middle figure}), and Algorithm~\ref{algorithm_max} (\emph{bottom figure})\vspace{-0.45cm}.}
\label{demon20}
\end{figure}

\vspace{0.3cm}

\begin{remark}
In Figure~\ref{demon20}, we observe that the operation of Algorithm~\ref{algorithm_max} allows agents, after a small amount of steps, to reach consensus to a value that is close but not necessarily equal to the average of the initial values. 
Eventually, this consensus value changes (around time step $90$) and becomes equal to the average of the initial values. 
This feature of Algorithm~\ref{algorithm_max} may be useful in situations in which the agents of a network need to coordinate their operations fast (i.e., reach a common decision) so the overall operation of the network is not disrupted greatly during the calculation of the exact average of the initial values (e.g., UAV flocking). 
However, similar behavior (i.e., consensus to a common value during the calculation of the average) can be observed in Algorithms~\ref{algorithm_prob} and \ref{algorithm1} if we modify them so they perform, along with their protocols, a ``leading mass'' max-voting (this effectively implies that the node that has the ``leading mass'' broadcasts is state variables to its out-neighbors).  
\end{remark}

Figures~\ref{compar20_dir_ring} and \ref{compar20_undir_ring} show what happens in the cases of a ring-shaped digraph of $20$ nodes and a ring-shaped undirected graph of $20$ nodes, respectively, in which the average of the nodes initial values is equal to $q = \dfrac{480}{20} = 24$. 
Again Algorithm~\ref{algorithm_max} appears to outperform Algorithms~\ref{algorithm_prob} and \ref{algorithm1}.

\vspace{-0.3cm}

\begin{figure} [ht]
\centering
\includegraphics[width=85mm]{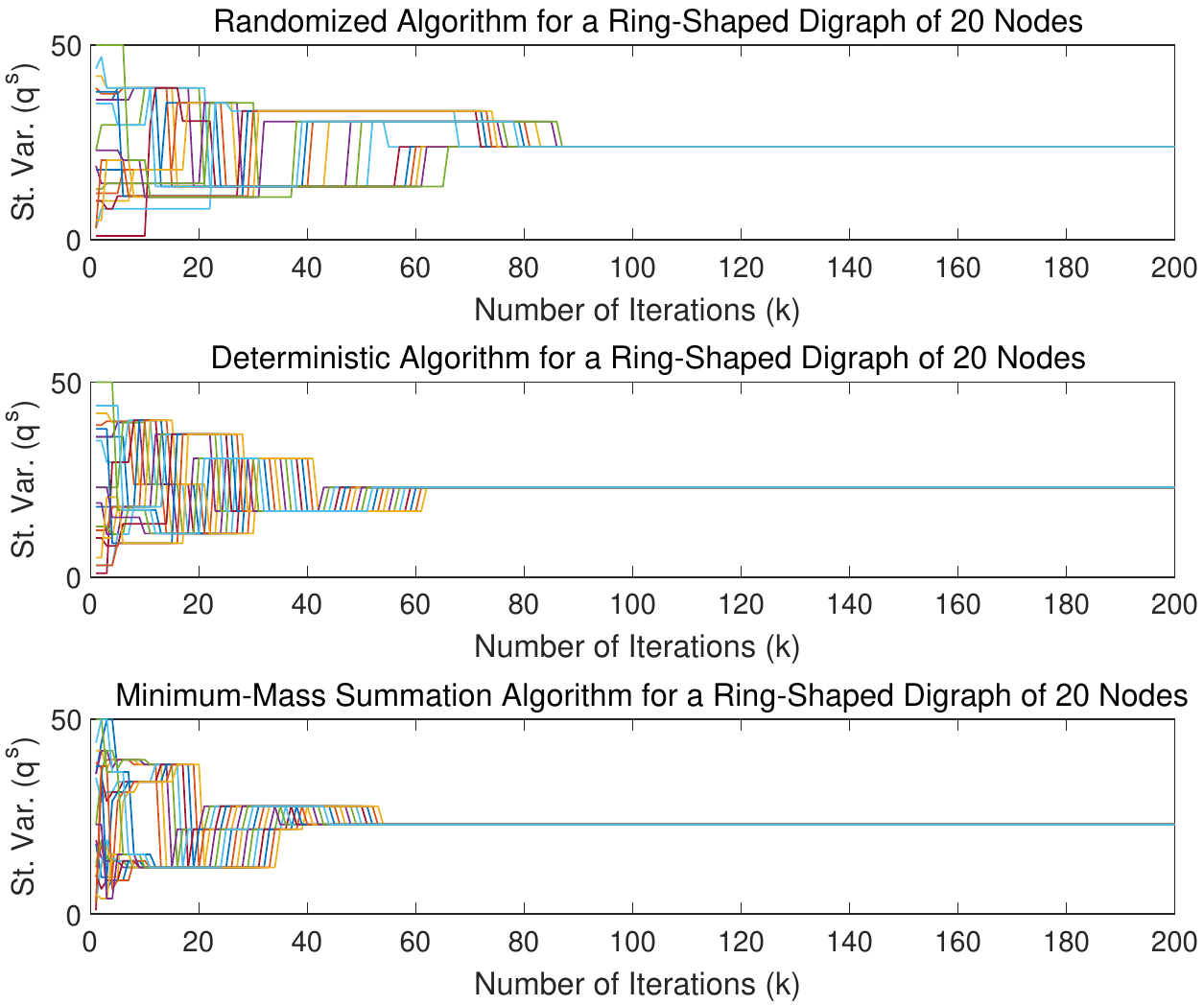}
\caption{Comparison between Algorithm~\ref{algorithm_prob}, Algorithm~\ref{algorithm1} and Algorithm~\ref{algorithm_max} for a ring-shaped digraph of $20$ nodes: 
Node state variables plotted against the number of iterations for Algorithm~\ref{algorithm_prob} (\emph{top figure}), Algorithm~\ref{algorithm1} (\emph{middle figure}), and Algorithm~\ref{algorithm_max} (\emph{bottom figure})\vspace{-0.45cm}.}
\label{compar20_dir_ring}
\end{figure}

\begin{figure} [ht]
\centering
\includegraphics[width=85mm]{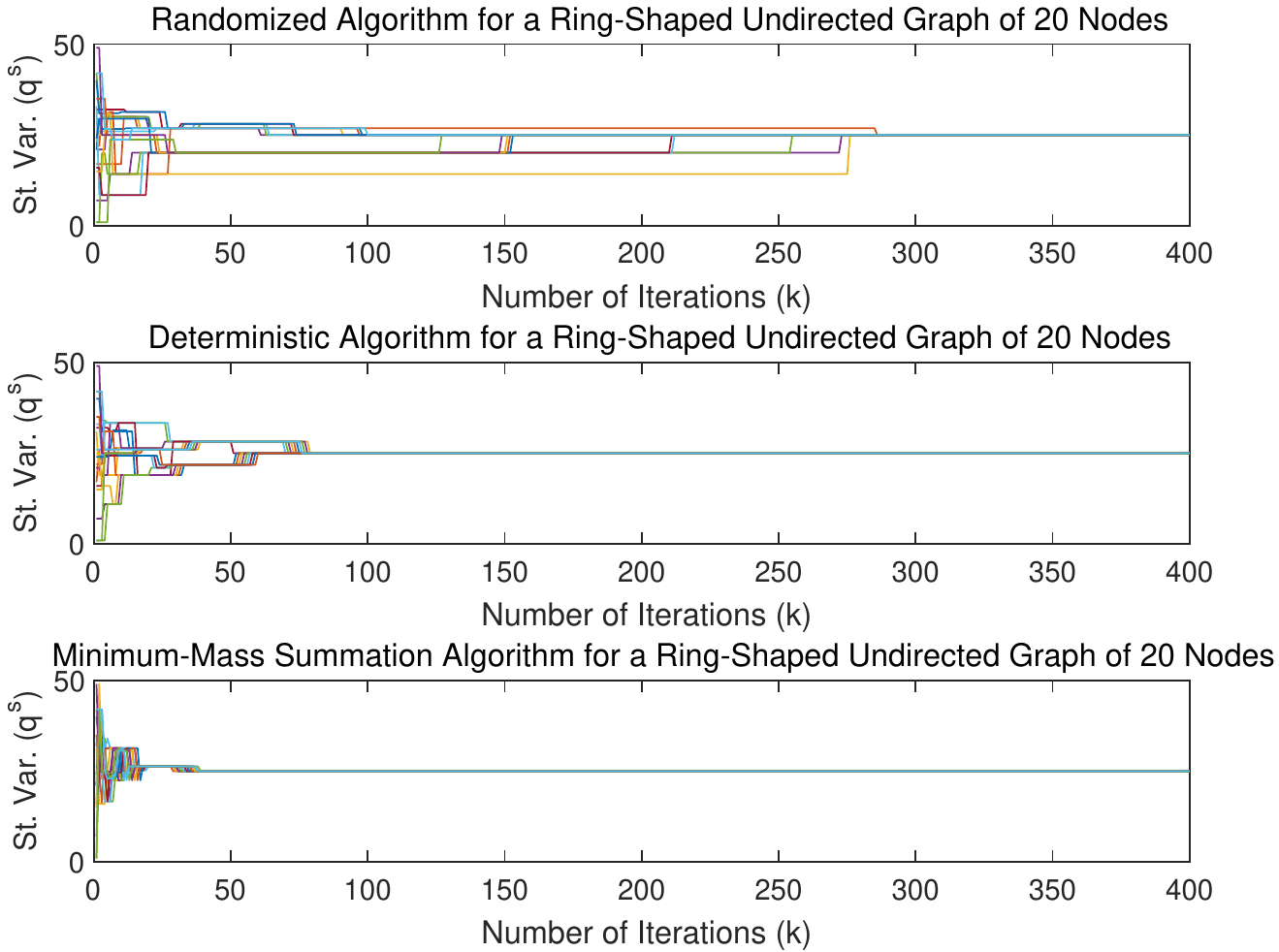}
\caption{Comparison between Algorithm~\ref{algorithm_prob}, Algorithm~\ref{algorithm1} and Algorithm~\ref{algorithm_max} for a ring-shaped undirected graph of $20$ nodes: 
Node state variables plotted against the number of iterations for Algorithm~\ref{algorithm_prob} (\emph{top figure}), Algorithm~\ref{algorithm1} (\emph{middle figure}), and Algorithm~\ref{algorithm_max} (\emph{bottom figure})\vspace{-0.45cm}.}
\label{compar20_undir_ring}
\end{figure}

\vspace{0.3cm}

Now, we compare the performance of the proposed algorithms against three other algorithms: (a) the quantized gossip algorithm presented in \cite{2007:Basar} in which, at each time step $k$, one edge\footnote{Note here that the algorithm presented in \cite{2007:Basar} requires the underlying graph to be undirected. For this reason, in Figure~\ref{comp20}, we consider, for the algorithm in \cite{2007:Basar}, the underlying graph to be undirected (i.e., if $(v_j, v_i) \in \mathcal{E}$ then $(v_i, v_j) \in \mathcal{E}$) while, for the algorithms in \cite{2011:Cai_Ishii, 2016:Chamie_Basar} we consider the underlying graph to be directed.} is selected at random, independently from earlier instants and the values of the nodes that the selected edge is incident on are updated, (b) the quantized asymmetric averaging algorithm presented in \cite{2011:Cai_Ishii} in which, at each time step $k$, one edge, say edge $(v_l, v_j)$, is selected at random and, node $v_j$ sends its state information and surplus and node $v_l$ performs updates 
over its own state and surplus values, (c) the distributed averaging algorithm with quantized communication presented in \cite{2016:Chamie_Basar} in which, at each time step $k$, each agent $v_j$ broadcasts a quantized version of its own state value towards its out-neighbors.

\begin{figure*}[ht]
\centering
\includegraphics[width=55mm]{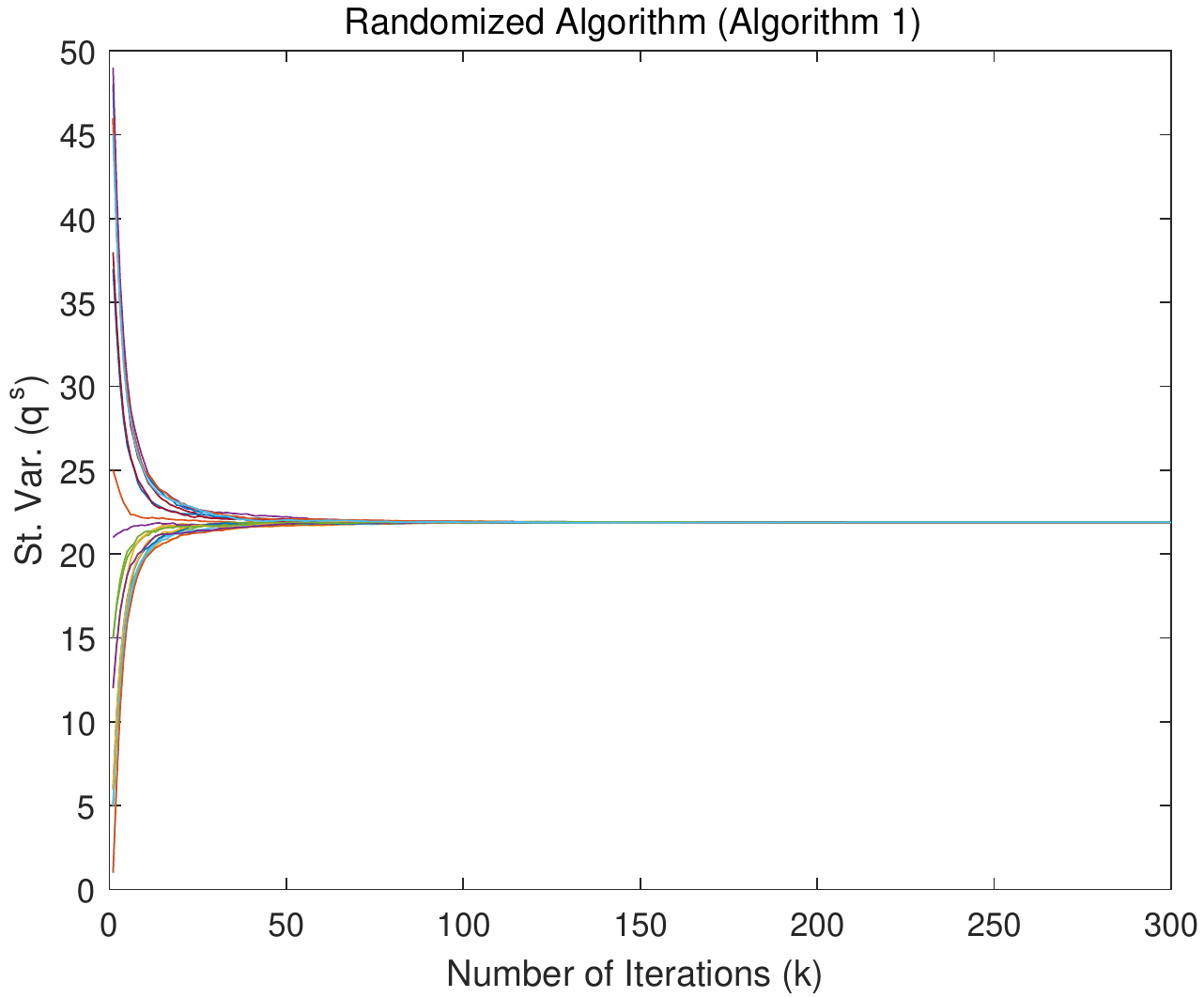}~~\hspace*{0.5cm}
\includegraphics[width=55mm]{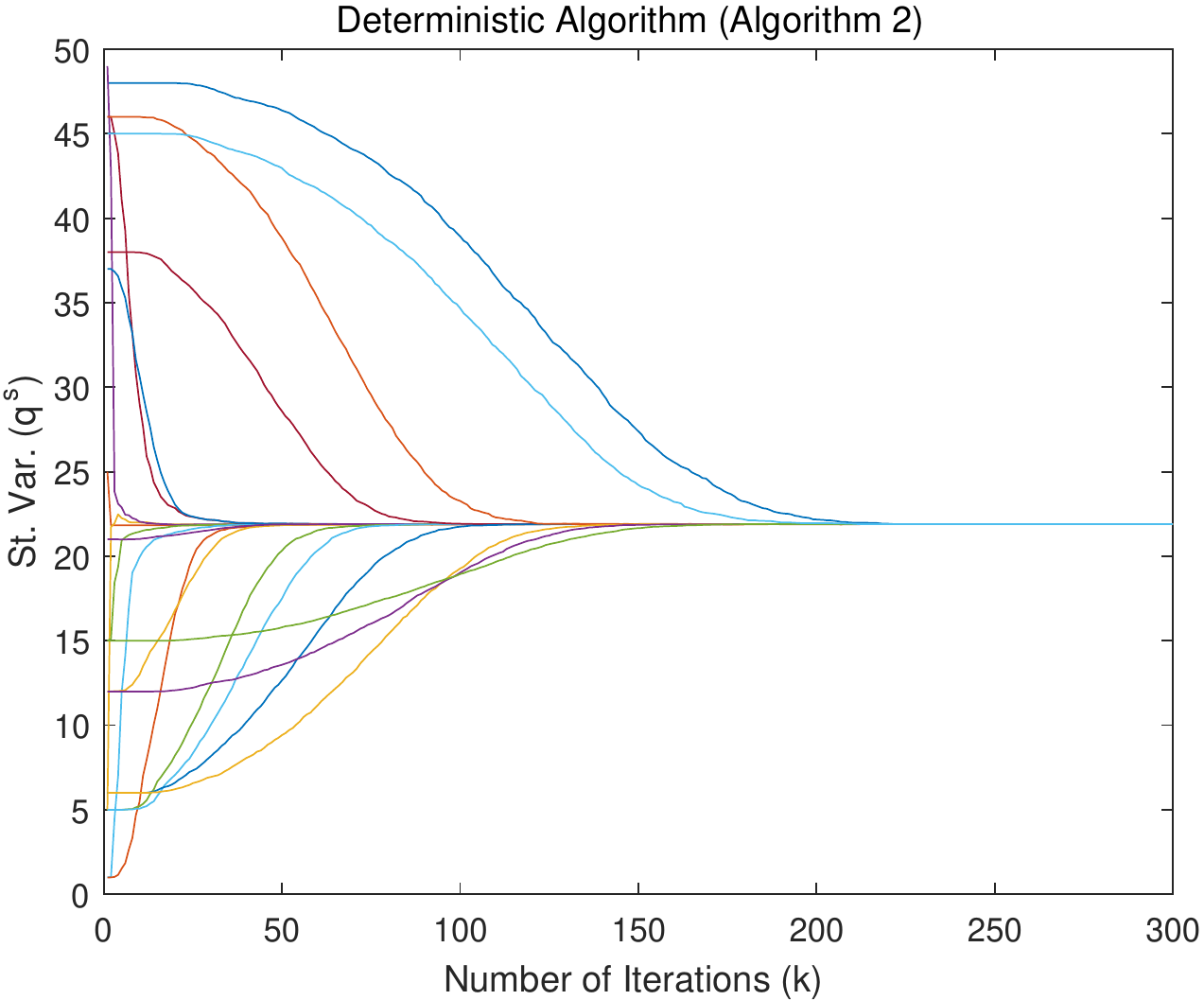}~~\hspace*{0.5cm}
\includegraphics[width=55mm]{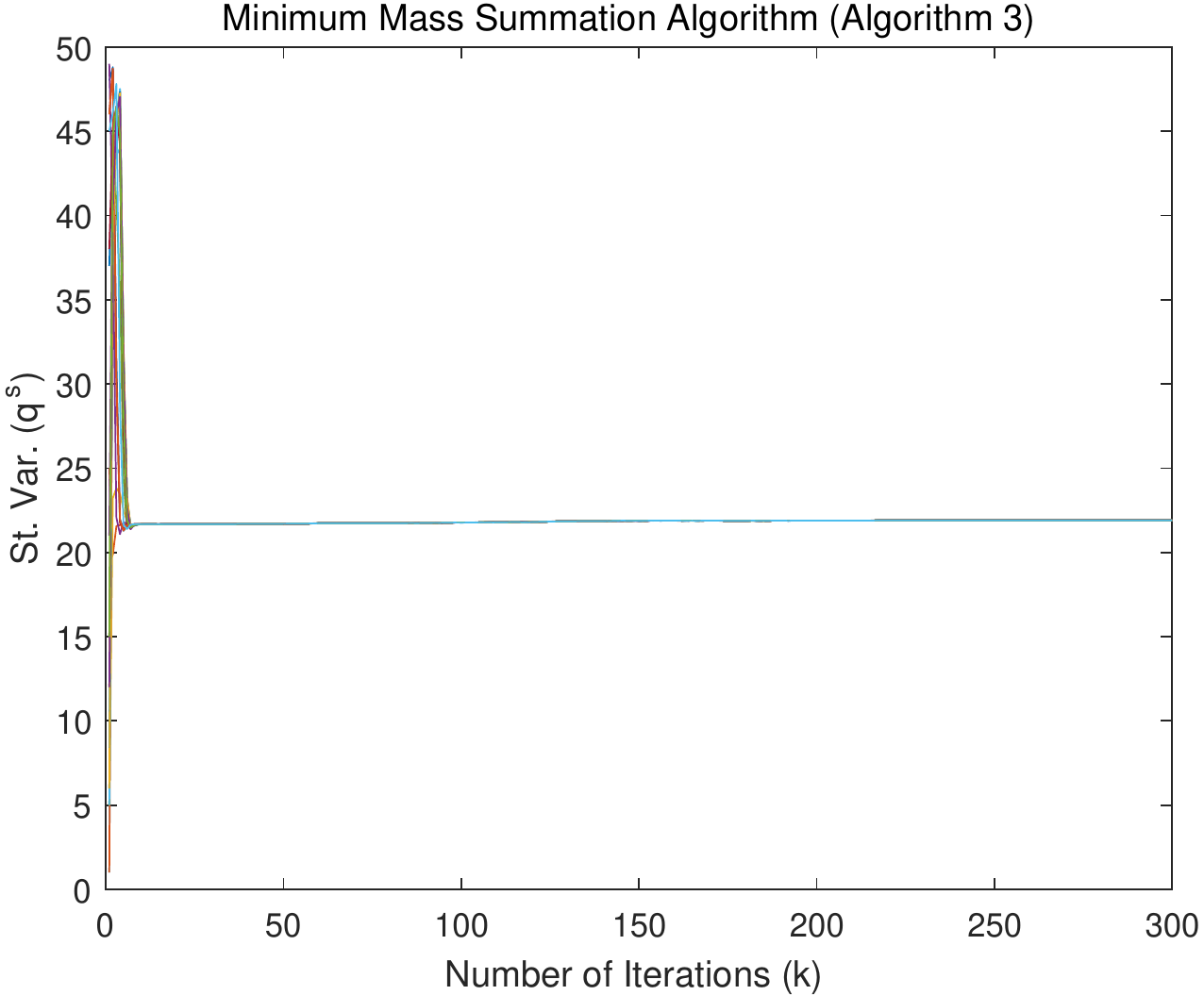} \\ 
\includegraphics[width=55mm]{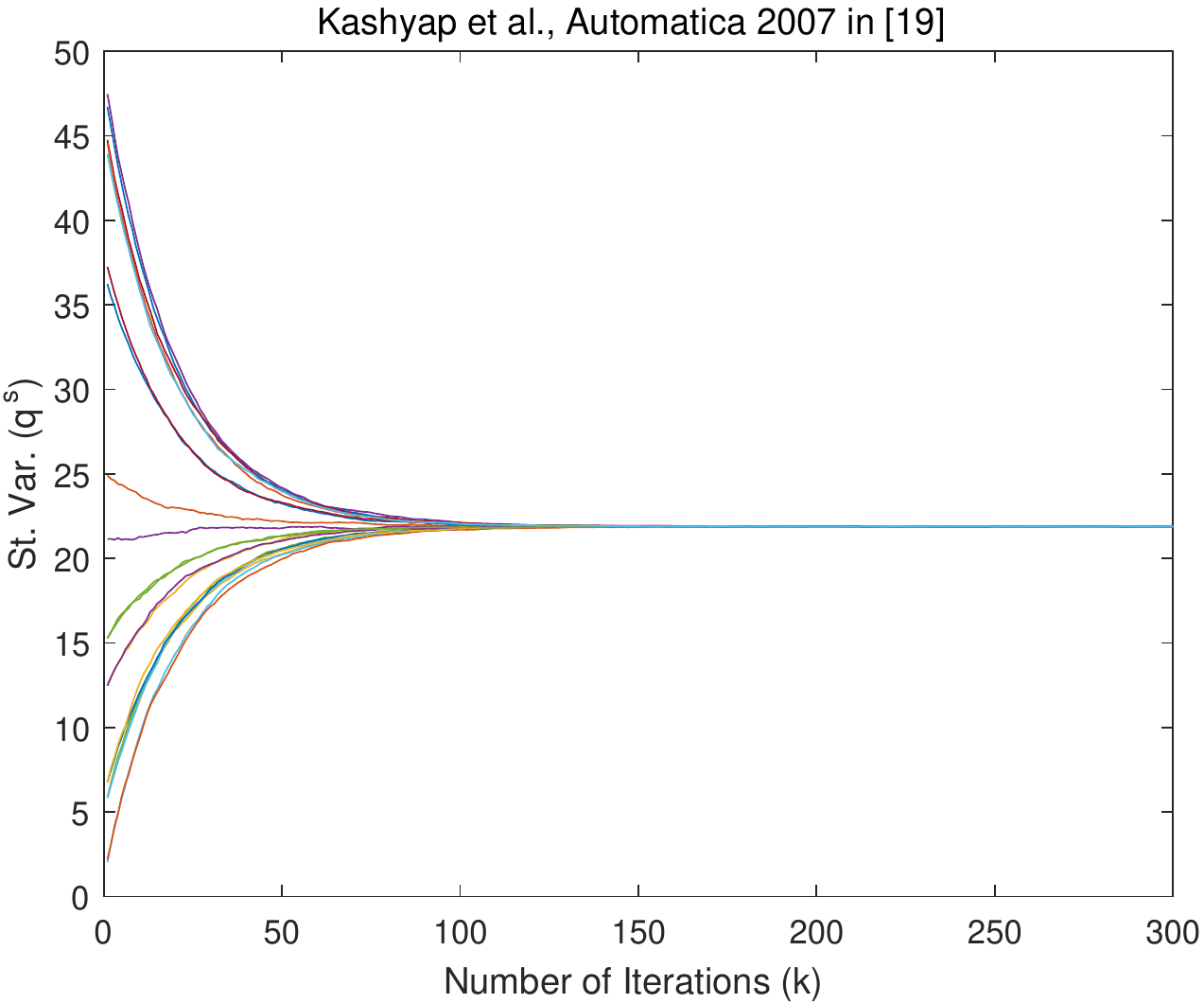}~~\hspace*{0.5cm}
\includegraphics[width=55mm]{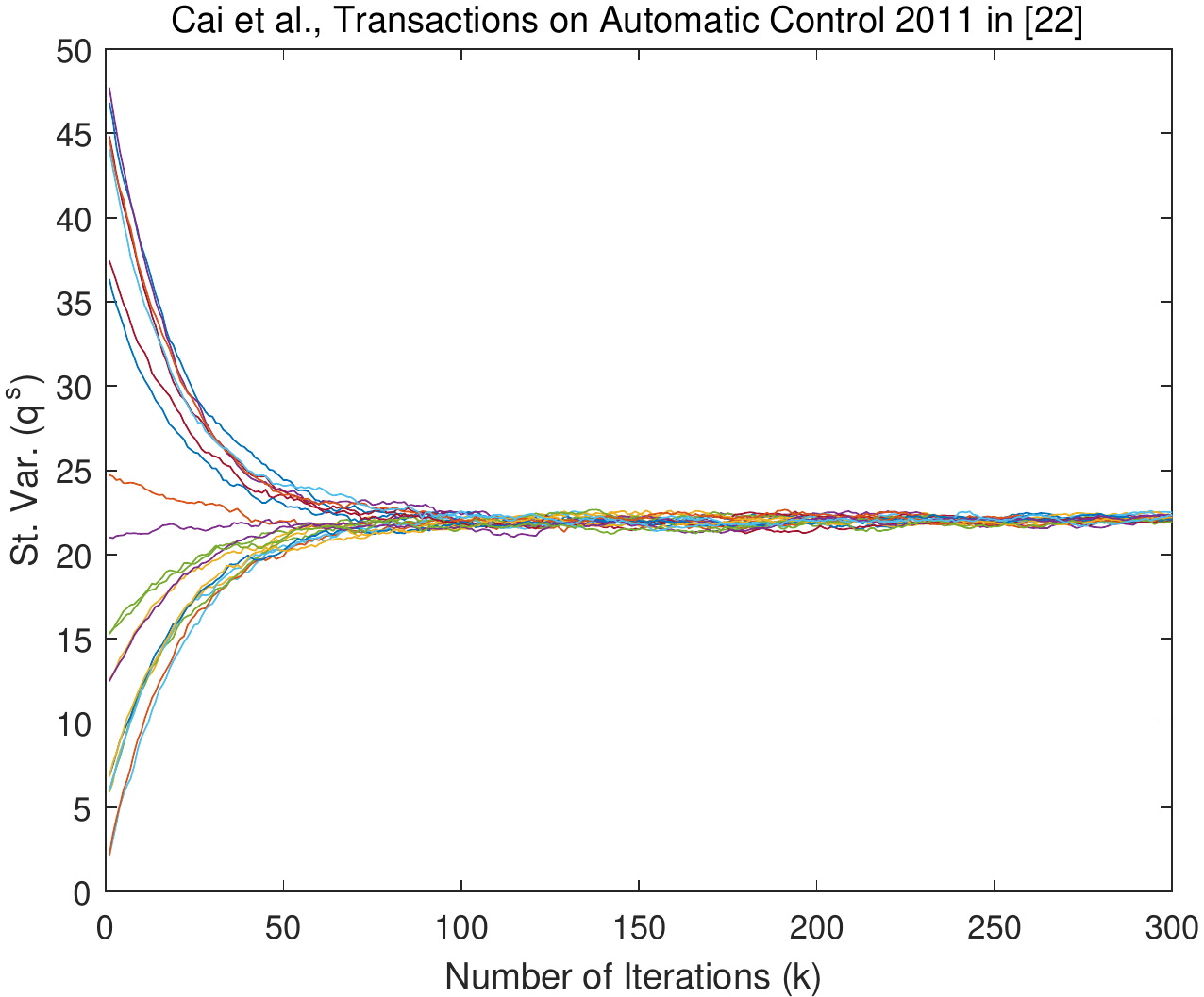}~~\hspace*{0.5cm}
\includegraphics[width=55mm]{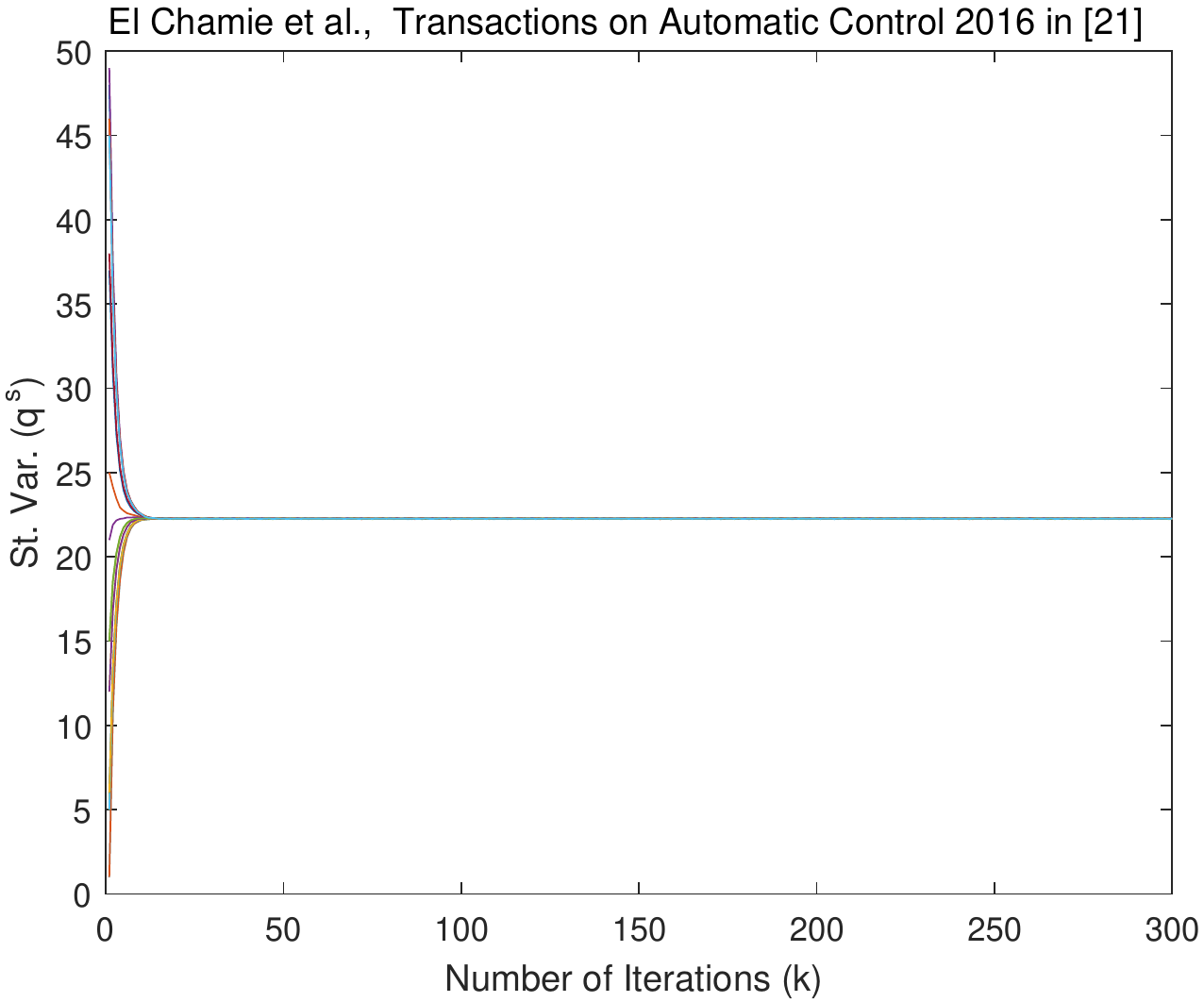}
\caption{Comparison between Algorithm~\ref{algorithm_prob}, Algorithm~\ref{algorithm1}, Algorithm~\ref{algorithm_max}, the quantized gossip algorithm presented in \cite{2007:Basar}, the quantized asymmetric averaging algorithm presented in \cite{2011:Cai_Ishii}, and the distributed averaging algorithm with quantized communication presented in \cite{2016:Chamie_Basar} for $1000$ random averaged digraphs of $20$ nodes each.}
\label{comp20}
\end{figure*}

Figure~\ref{comp20} presents a study of the case of $1000$ digraphs of $20$ nodes each, in which the average of the nodes initial values is equal to $q = \dfrac{440}{20} = 22$. 
The results shown are the average results over $1000$ graphs. 
The top of Figure~\ref{comp20} shows that Algorithm~\ref{algorithm_prob} and Algorithm~\ref{algorithm_max} outperform the quantized gossip algorithm presented in \cite{2007:Basar} (for which the underlying graph is undirected), while Algorithm~\ref{algorithm1} requires a relatively larger number of iterations to converge. 
Furthermore, we can see that Algorithm~\ref{algorithm_prob}, Algorithm~\ref{algorithm1} and Algorithm~\ref{algorithm_max} outperform the quantized asymmetric averaging algorithm presented in \cite{2011:Cai_Ishii}. 
However, the distributed averaging algorithm presented in \cite{2016:Chamie_Basar} is able to outperform all the aforementioned algorithms due to the fact that nodes can process real numbers and use a set of weights that form a doubly stochastic matrix (which guarantees convergence proportional to the second largest magnitude of its eigenvalues). 
Notice that obtaining a set of weights that forms a doubly stochastic matrix is an easy task in undirected graphs but becomes challenging in directed graphs \cite{2018:BOOK_Hadj}.

\section{CONCLUDING REMARKS}\label{future}

In this work, we have considered the quantized average consensus problem for a distributed system that forms a directed graph in which the agents can exchange quantized information in a distributed fashion. 
Quantized average consensus plays a key role in a number of applications, which aim at more efficient usage of network resources. 
We proposed three distributed algorithms for solving the quantized average consensus problem. 
For the first {\em probabilistic} distributed algorithm, we showed that each agent achieves quantized average consensus with probability one. 
For the second and third {\em deterministic} distributed algorithms, we showed that quantized average consensus is achieved after a finite number of iterations that we explicitly bounded; for the latter, we also showed that once quantized average consensus is achieved, transmissions are ceased from each agent. 
To the best of our knowledge, these are the first {\em deterministic} algorithms which allow convergence to the quantized average of the initial values after a finite number of time steps without any specific requirements regarding the network that describes the underlying communication topology (see \cite{2016:Chamie_Basar}).

In the future, we plan to extend the proposed algorithms over other types of quantized consensus, such as quantized average consensus on networks with time-varying topologies or networks with bounded or unbounded (packet drops) transmission delays over the communication links. 
Furthermore, we plan to design distributed strategies under which every agent in the network will be able to determine whether quantized average consensus has been reached (and thus proceed to execute more complicated control or coordination tasks). 
Finally, we also plan to design algorithms that achieve quantized average consensus by preserving the privacy of the participating agents.

\vspace{-0.3cm}

% ------------------------------------------------------------------------------
% Bibliography
% ------------------------------------------------------------------------------
\bibliographystyle{IEEEtran}
\bibliography{bibliografia_consensus}

\vspace{-0.3cm}

\begin{IEEEbiography}[{\includegraphics[width=1in,height=1.25in,clip,keepaspectratio]{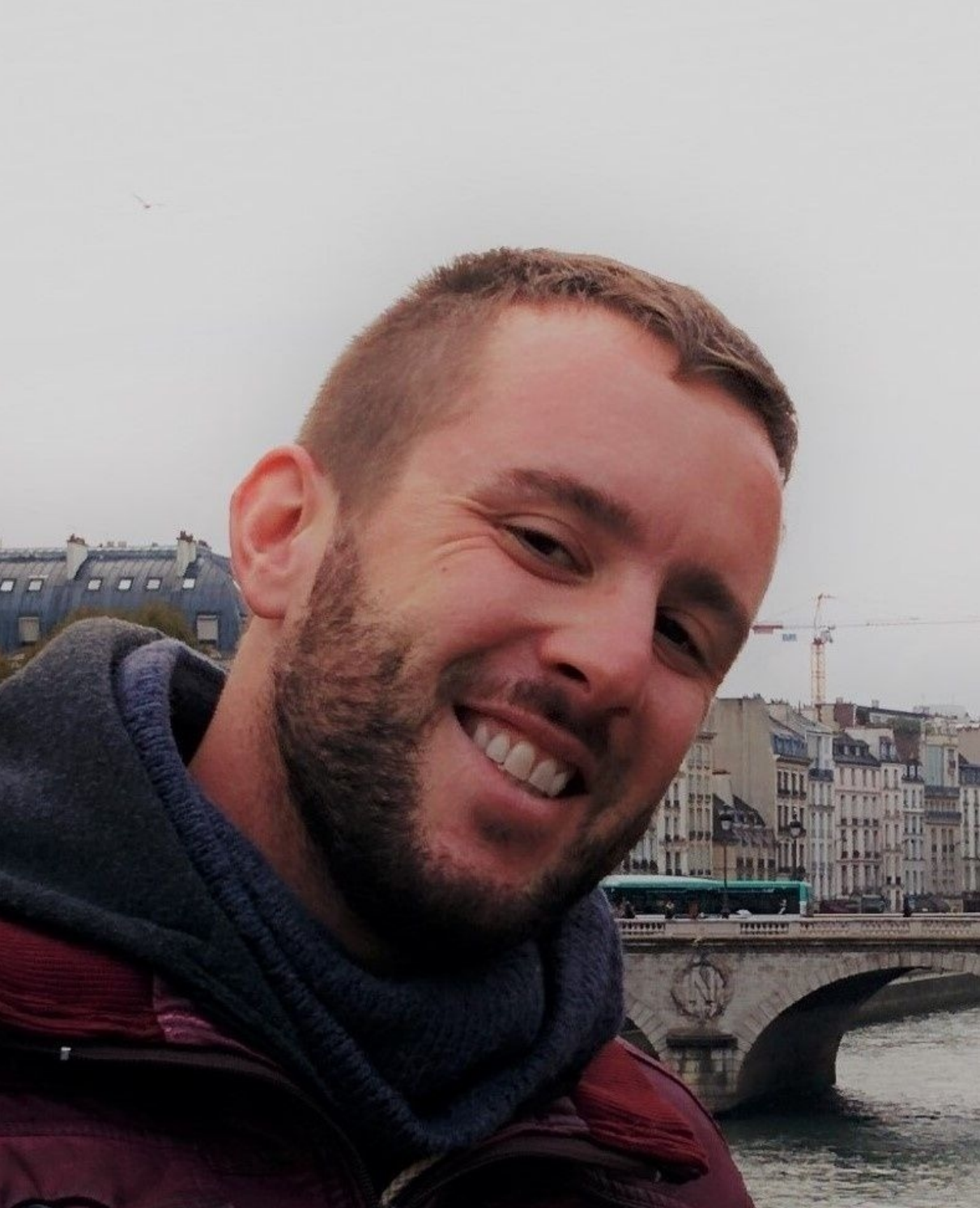}}]{Apostolos I. Rikos} (M'16) received the B.Sc., M.Sc and Ph.D. degrees in Electrical Engineering from the Department of Electrical and Computer Engineering, University of Cyprus in 2010, 2012 and 2018 respectively.
In 2018, he joined the KIOS Research and Innovation Centre of Excellence in Cyprus. 
Since February 2020, he has been a postdoctoral researcher at the Automatic Control Department of KTH Royal Institute of Technology. 
His research interests are in the area of distributed systems, coordination and control of networks of autonomous agents, sensor networks, stochastic processes, optimization and graph theory.
\end{IEEEbiography}

\vspace{-0.2cm}

\begin{IEEEbiography}[{\includegraphics[width=1in,height=1.25in,clip,keepaspectratio]{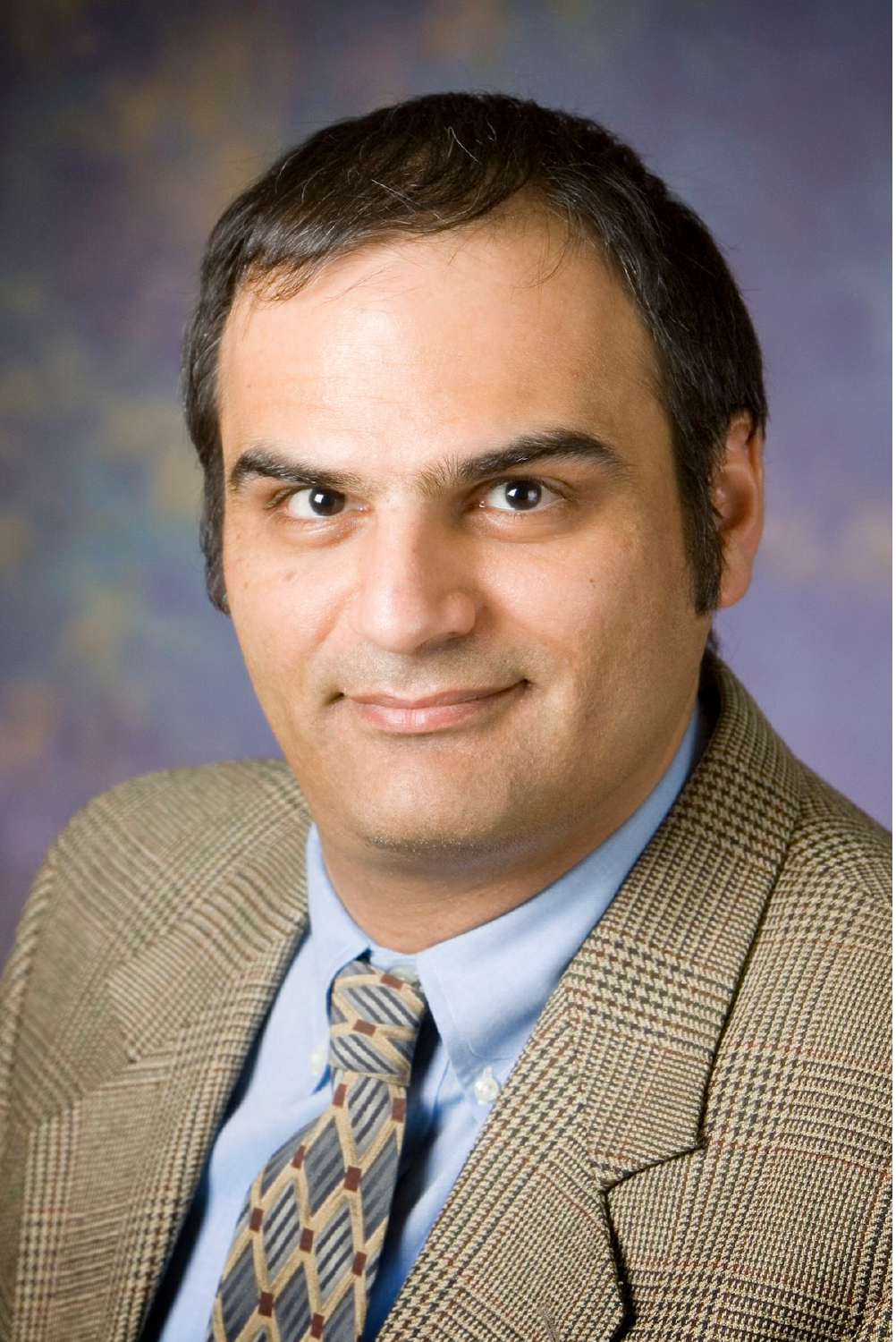}}]{Christoforos N. Hadjicostis}
(M'99, SM'05) received the S.B. degrees in electrical engineering, computer science and engineering, and in mathematics, the M.Eng. degree in electrical engineering and computer science in 1995, and the Ph.D. degree in electrical engineering and computer science in 1999, all from Massachusetts Institute of Technology, Cambridge. In 1999, he joined the Faculty at the University of Illinois at Urbana--Champaign, where he served as Assistant and then Associate Professor with the Department of Electrical and Computer Engineering, the Coordinated Science Laboratory, and the Information Trust Institute. Since 2007, he has been with the Department of Electrical and Computer Engineering, University of Cyprus, where he is currently Professor. His research focuses on fault diagnosis and tolerance in distributed dynamic systems, error control coding, monitoring, diagnosis and control of large-scale discrete-event systems, and applications to network security, anomaly detection, energy distribution systems, medical diagnosis, biosequencing, and genetic regulatory models. He currently serves as Departmental Editor of the Journal of Discrete Event Systems and as Associate Editor of Automatica and the Journal of Nonlinear Analysis of Hybrid Systems. In the past, he had served as Associate of IEEE Transactions on Automatic Control, IEEE Transactions on Automation Science and Engineering, IEEE Transactions on Control Systems Technology, and IEEE Transactions on Circuits and Systems~I.
%(M'99--SM'05) received the B.S. degrees in electrical engineering, computer science and engineering, and mathematcs, and the M.Eng. and Ph.D. degrees in electrical engineering and computer science from the Massachusetts Institute of Technology, Cambridge, MA, USA, in 1993, 1995, and 1999, respectively.
%In 1999, he joined the Faculty at the University of Illinois at Urbana-Champaign, where he served as Assistant and then Associate Professor with the Department of Electrical and Computer Engineering,
%the Coordinated Science Laboratory, and the Information Trust Institute. Since 2007, he has been with the Department of Electrical and Computer Engineering, University of Cyprus, Nicosia , Cyprus. His research focuses on fault diagnosis and tolerance in distributed dynamic systems; error-control coding; monitoring, diagnosis, and control of large-scale discrete-event systems; and applications to network security, anomaly detection, energy distribution systems, medical diagnosis, biosequencing, and genetic regulatory models.
\end{IEEEbiography}

\end{document}